\documentclass[pdflatex,sn-mathphys-num]{sn-jnl}

\usepackage{natbib}  
\usepackage{multirow}%
\usepackage{amsmath,amssymb,amsfonts}%
\usepackage{amsthm}%
\usepackage{mathrsfs}%
\usepackage[title]{appendix}%
\usepackage{xcolor}%
\usepackage{textcomp}%
\usepackage{manyfoot}%
\usepackage{booktabs}%
\usepackage{algorithm}%
\usepackage{algorithmicx}%
\usepackage{algpseudocode}%
\usepackage{listings}%
\usepackage{commath}
\usepackage{cancel}
\usepackage{enumerate}
\usepackage{comment}
\usepackage{bbm}
\usepackage{subfig}
\usepackage{tabularx}
\usepackage{graphicx}
\usepackage[export]{adjustbox}
\usepackage{wrapfig}

\usepackage[font={small}]{caption}

\newtheorem{theorem}{Theorem}[section]
\newtheorem{proposition}[theorem]{Proposition}%
\newtheorem{remark}{Remark}%

\newtheorem{definition}{Definition}%

\newtheorem{assumption}{Assumption}%

\newtheorem{corollary}[theorem]{Corollary}%
\newtheorem{lemma}[theorem]{Lemma}%

\newcommand{\mN}{\mathcal{N}}
\newcommand{\mB}{\mathcal{B}}

\newcommand{\mF}{\mathcal{F}}
\newcommand{\mM}{\mathcal{M}}
\newcommand{\mP}{\mathcal{P}}

\newcommand{\mX}{\mathcal{X}}
\newcommand{\mY}{\mathcal{Y}}

\newcommand{\sd}{{\sf d}}
\newcommand{\rw}{\rightarrow}

\newcommand{\mbP}{\mathbb{P}}
\newcommand{\mbR}{\mathbb{R}}

\newcommand{\mbN}{\mathbb{N}}

\newcommand{\beq}{\begin{equation}}
\newcommand{\eeq}{\end{equation}}
\newcommand{\beqa}{\begin{eqnarray}}
\newcommand{\eeqa}{\end{eqnarray}}
\newcommand{\nn}{\nonumber}

\newcommand{\stcomp}[1]{{#1}^{\mathsf{c}}} 

\usepackage[textwidth=1.8cm, textsize=scriptsize]{todonotes}

\renewcommand{\top}{\mathsf{T}}

\begin{document}

\title[Nudging state-space models]{Nudging state-space models for Bayesian filtering under misspecified dynamics
}


\author[1]{\fnm{Fabián} \sur{González}}\email{omgonzal@math.uc3m.es}

\author[2]{\fnm{O. Deniz } \sur{Akyildiz}}\email{deniz.akyildiz@imperial.ac.uk}

\author[2]{\fnm{Dan} \sur{Crisan}}\email{d.crisan@imperial.ac.uk}

\author[1,3]{\fnm{Joaquín} \sur{Míguez}}\email{joaquin.miguez@uc3m.es}

\affil[1]{\orgdiv{Department of Signal Theory and Communications}, \orgname{Universidad Carlos III de Madrid (ROR: \url{https://ror.org/03ths8210})}, \orgaddress{\street{ Avda. de la Universidad, 30}, \city{ Leganés}, \postcode{28911}, \state{Madrid}, \country{Spain}}}

\affil[2]{\orgdiv{Department of Mathematics}, \orgname{Imperial College London}, \orgaddress{\street{180 Queen's Gate}, \city{London}, \postcode{SW7 2AZ}, 
\country{United Kingdom}}}

\affil[3]{\orgname{Instituto de Investigaci\'on Sanitaria Gregorio Mara\~n\'on}, \orgaddress{\street{Calle Dr. Esquerdo, 46}, \city{Madrid}, \postcode{28007}, \country{Spain}}}


\abstract{Nudging is a popular algorithmic strategy in numerical filtering to deal with the problem of inference in high-dimensional dynamical systems. We demonstrate in this paper that general nudging techniques can also tackle another crucial statistical problem in filtering, namely the misspecification of the transition kernel. Specifically, we rely on the formulation of nudging as a general operation increasing the likelihood and prove analytically that, when applied carefully, nudging techniques implicitly define state-space models that have higher marginal likelihoods for a given (fixed) sequence of observations. This provides a theoretical justification of nudging techniques as data-informed algorithmic modifications of state-space models to obtain robust models under misspecified dynamics. To demonstrate the use of nudging, we provide numerical experiments on linear Gaussian state-space models and a stochastic Lorenz 63 model with misspecified dynamics and show that nudging offers a robust filtering strategy for these cases.}




\maketitle

\section{Introduction}\label{sIntroduction}

\subsection{State space models}

State-space models (SSMs) are key building blocks in many applications in signal processing, machine learning, weather forecasting, finance, object tracking, ecology, and many other fields \cite{triantafyllopoulos2021bayesian}. These models are used to represent the dynamics of a system, where the \textit{system state} evolves over time according to a Markov transition kernel and the available \textit{observations} (data) are related to the system state by a likelihood function. The main statistical goal in SSMs is to infer the state of the system given a sequence of observations, a problem known as filtering \cite{anderson2005optimal}.

Formally, we represent the state of the SSM by a Markov chain $\{X_t\}_{t \ge 0}$ described as follows. The initial state $X_0$ is a random variable (r.v.) with probability law $\pi_0$ and, at any time $t \ge 1$, the dynamics of the transition from $X_{t-1}$ to $X_t$ is modelled by a Markov kernel $K_t(x_{t-1}, \sd x_t)$. The sequence of observations is denoted by $\{Y_t\}_{t \ge 1}$ and the relationship between the state $X_t$ and the observation $Y_t$ is modelled by a conditional probability density function (pdf) $p_t(y_t|X_t=x_t)$. Since in practice the observations are given, $Y_t=y_t$ for $t \ge 1$, the latter relationship is usually given in terms of a likelihood function $g_t(x_t) \propto p_t(y_t|X_t=x_t)$. With these elements, the conditional probability law of the state $X_t$ given the data $Y_{1:t} = y_{1:t} := \{ y_1, \ldots, y_t \}$ can be constructed recursively via the Chapman-Kolmogorov equation and Bayes' theorem (see, e.g., \cite{DelMoral04,Bain09}) and we denote it as $\pi_t$. The conditional law $\pi_t$ is often termed the optimal, or Bayesian, filter.

The optimal filter $\pi_t$ can only be computed exactly in a few specific cases. The most relevant one is the scenario where both the Markov kernels $K_t$ and the likelihoods $g_t$ correspond to linear relationships and Gaussian noise. Under such assumptions, $\pi_t$ is Gaussian and its mean and covariance matrix can be computed recursively via the Kalman filter (KF) algorithm \cite{kalman1960new}. In most practical applications, however, the optimal filter $\pi_t$ can only be approximated numerically using nonlinear KFs, particle filters (PFs) or other approximation methods \cite{Ristic04,Bain09,sarkka2023bayesian}.



\subsection{Model misspecification}

One of the main challenges in Bayesian filtering is model misspecification, which occurs when the chosen family of transition models, $\{K_t\}_{t \ge 1}$, the likelihood functions, $\{g_t\}_{t \ge 1}$, or both, fail to represent the statistical properties of the real-world system of interest with sufficient accuracy. Model misspecification is a long-standing problem in Bayesian filtering and it has been studied from different viewpoints in the literature, including outlier detection, robust filtering, parameter estimation, and the so-called \textit{nudging} techniques.

Outlier detection \cite{blazquez2021review} is, perhaps, the simplest way to manage observations which are in poor agreement with the assumed SSM. In the context of filtering, typical outlier detection schemes approximate the predictive distribution of the upcoming observation $Y_t$. Then, when the actual observation is collected, a statistical test can be run to determine whether the observed data $y_t$ is compatible with the predicted distribution for $Y_t$. If the test indicates that the observation is anomalous (i.e., it is an outlier with respect to (w.r.t.) the predicted distribution) then the data $y_t$ can either be discarded or be processed using a {\em robust} procedure that mitigates the effect of the outlying data on the filter update. Many methods for outlier detection and robust filtering have been proposed for Kalman-based filters \cite{yatawara1991kalman,xie1994robust,agamennoni2011outlier,petersen2012robust,javanfar2023measurement,truzman2024outlier} or PFs \cite{Maiz12,Vazquez17,boustati2020generalised,zhang2023outlier}. A fundamental problem with these approaches is that anomalous data are handled as detrimental and uninformative, under the assumption that they have not been generated by the system of interest. Very often, however, a genuine observation from the system of interest may appear as an outlier because of the misspecification of the SSM. By discarding or mitigating this observation, relevant information is wasted and model errors are reinforced.

Another classical strategy to account for modelling uncertainty is to choose not \textit{one} SSM but a parametric family of SSMs indexed by a (possibly multidimensional) parameter $\theta$. When a sequence of observations becomes available, the model is calibrated by tuning the parameter $\theta$ to the data according to some statistical criterion. Maximum likelihood estimation methods have been proposed, both offline \cite{jensen1999asymptotic,Olsson08} and online \cite{LeGland97,Tadic10}, as well as Bayesian estimation methods. The latter include algorithms such as particle Markov chain Monte Carlo (MCMC) \cite{Andrieu10,lindsten2012ancestor,dahlin2015particle}, iterated batch importance sampling (IBIS) \cite{Chopin02} or sequential Monte Carlo square (SMC$^2$) \cite{Chopin12}, and recursive algorithms like the nested PF \cite{Crisan18bernoulli} and its Kalman-based approximations \cite{Perez-Vieites18,Perez-Vieites21}. While parameter estimation methods are an indispensable toolbox for practical applications, they do not provide a complete solution to the model misspecification problem. Indeed, the parametric family of SSMs may not be flexible enough the represent the features of the system of interest, no matter the choice of $\theta$. For example, a parametric class of linear models can be expected to fail to represent a system that displays non-negligible nonlinear features in its dynamics.

Several techniques collectively known as \textit{nudging} have been devised to mitigate model misspecification \cite{reich2015probabilistic}. Nudging methods are designed to steer (or \textit{nudge}) a model towards the observed data over time by adding a (small) corrective term to the model dynamics. The goal is to make the model follow observed values more closely without breaking down its original dynamics. Such methods have been particularly popular within the data assimilation community \cite{Stroud2010ensemble,reich2015probabilistic,Law15}. Nudging can be used as a stand-alone data assimilation method \cite{lakshmivarahan2013nudging,farhat2017continuous,desamsetti2019efficient} but it is often combined with ensemble KFs \cite{luo2012ensemble,lei2012hybridI,lei2012hybridII} or PFs \cite{dubinkina2013assessment,lingala2014optimal,akyildiz2020nudging}. In the context of particle filtering, nudging has been interpreted either as a tool to design efficient proposals \cite{dubinkina2013assessment,lingala2014optimal} or as a modification of the sampling scheme \cite{akyildiz2020nudging}. A similar approach has also been used in simulation-based Bayesian inference and machine learning, typically by incorporating additional parameters that can be learned from data in order to mitigate model errors and improve robustness \cite{frazier2021robust,ward2022robust,kelly2024misspecification}. Here, we advocate nudging as a flexible tool to compensate model specification: the correction term can be constructed in many different ways, by applying different criteria, and it can be further combined with parameter estimation methods and outlier rejection techniques if needed.


\subsection{Contributions}

In this paper we adopt a viewpoint of nudging as a data-informed modification of the kernels $\{K_t\}_{t \ge 1}$ of the SSM, rather than a tweak of the filtering algorithms. In particular, let $\mM$ denote the original SSM. We introduce a broad family of \textit{nudging maps} $(\alpha_t)_{t \ge 1}$ which, given the available observations $\{y_t\}_{t \ge 1}$, yield a sequence of modified (nudged) kernels $\{K_t^\alpha\}_{t \ge 1}$. These kernels, in turn, characterise a modified SSM, denoted $\mM^\alpha$, which is therefore different from the original $\mM$. We investigate the relative agreement of the two models, $\mM$ and $\mM^\alpha$, with a given data set $y_{1:T}$. This agreement is quantified by means of the marginal likelihoods, or Bayesian model evidence, of the two SSMs \cite{Knuth15evidence}. The key contributions and findings of this research are outlined below.  
\begin{itemize}
\item We describe a general nudging methodology that consists of a data-driven modification of the kernels $\{K_t\}_{t \ge 1}$ in the SSM. This modification is defined by a parametric nudging transformation that satisfies some regularity conditions and admits many different practical implementations.

\item For a given set of observations $y_{1:T}$, and under mild assumptions on the original SSM $\mM$, we prove that the proposed nudging methodology can yield a modified model $\mM^\alpha$ that attains a higher marginal likelihood than the base model $\mM$. 

In particular, when the original model $\mM$ is indexed by a vector of parameters $\theta$, i.e., $\mM \equiv \mM_\theta$, we prove that the nudged model $\mM_\theta^\alpha$ can attain a marginal likelihood that (a) is higher than the marginal likelihood of the model $\mM_\theta$, with the same parameters $\theta$, and (b) lies in a neighbourhood of the marginal likelihood attained by model $\mM_{\theta_*}$, where $\theta_*$ is the maximum likelihood estimator of the parameters.

\item We describe a specific class of nudging transformations that rely on the ability to compute the gradient of the log-likelihood function $\log g_t$ of the original model $\mM$. We prove that the  theoretical guarantees obtained for the general parametric transformations also hold for the proposed gradient-based nudging. This version of nudging is straightforward to implement when $g_t$ belongs to the exponential family (e.g., if the observation noise is additive and Gaussian). Note that there are also standard numerical methods that can be used to approximate the gradient of $\log g_t$ when the likelihood is analytically intractable \cite{del2015uniform,jasra2021unbiased}.

\item We apply the proposed methodology, with gradient-based nudging transformations, to the class of linear-Gaussian SSMs and explicitly obtain a nudged version of the KF (i.e., a KF for the nudged model $\mM^\alpha$). Then, we identify explicit conditions on the original SSM $\mM$ that, when satisfied, guarantee that the nudged KF yields a higher marginal likelihood than the original algorithm.

\item Finally, we demonstrate the application of the methodology, and illustrate the theoretical results numerically for two models. The first one is a four-dimensional linear Gaussian model, while the second one is a stochastic Lorenz 63 model with partial observations. We show numerically, in both examples, that the proposed gradient-based nudging methodology can yield an increased marginal likelihood and compensate for errors in the model parameters. 
\end{itemize}

%
\subsection{Outline of the paper}

We conclude this introduction with brief summary of the notation used throughout the paper, presented in Section \ref{ssNotation}. In Section \ref{sec2} we provide a formal description of the SSMs of interest, the optimal Bayesian filter and the Bayesian model evidence. The proposed nudging methodology is introduced in Section \ref{sNudging}, which also contains the main theoretical results. Computer simulation results for a linear-Gaussian model and a stochastic Lorenz 63 model are presented in Section \ref{sNumerical}. Section \ref{sConclusions} contains a summary of the main results and some concluding remarks. The proofs of the main theorems, as well as some additional technical results, are presented in Appendices \ref{App1}--\ref{A:PGAN}. 


%
\subsection{Summary of notation} \label{ssNotation}

\begin{itemize}
    \item Sets, measures, and integrals:
    \begin{itemize}
        \item[-] $\mB (S)$ is the $\sigma$ -algebra of Borel subsets of $S \subseteq  \mbR^d$.
\item[-] $\mP(S) := \{ \mu  : \mB (S) \mapsto  [0, 1]$ and $\mu (S) = 1\}$  is the set of probability measures over
$\mB (S)$.
\item[-] $ \mu (f) := 
\int f d\mu$  is the integral of a Borel measurable function $f : S \mapsto  \mbR$ w.r.t. the
measure $\mu  \in  \mP (S)$.
\item[-] The indicator function on a set $S$ is denoted by $\mathbbm{1}_S(x)$. Given a measure $\mu$  and a
set $S$, we equivalently denote $\mu (S) := \mu( \mathbbm{1}_S)$.
\item[-] Let $A$ be a subset of a reference space $\mX  \subseteq  \mbR^{d_x}$. The complement of $A$ w.r.t. $\mX$  is
denoted by $\stcomp{A} := \mX \setminus A$.
\item [-] Let $\mu$ be a finite measure over $(\mX,\mB(\mX))$ (i.e., $\mu(\mX)< \infty$). The total variation norm of $\mu$ is
$$\norm{\mu}_{TV}:=\abs{
        \sup_{F\in \mB (\mX)} \mu(F)-\inf_{F\in \mB (\mX)} \mu(F)}.$$
\end{itemize}
\item Functions and sequences:
    \begin{itemize}
    \item [-] $B(S)$ is the set of bounded real functions over $S$. Given $f \in B(S)$, we denote
$$\| f\|_{\infty}  := \sup_{s\in S}
| f(s)|<\infty.$$
\item [-] We use a subscript notation for subsequences, namely $x_{t_1:t_n}:=\{x_{t_1},...,x_{t_n}\}.$
    \end{itemize}
\item Real r.v.'s on a probability space $(\Omega,\mF,\mbP)$ are denoted by capital letters (e.g., $Z:\Omega \mapsto \mbR^d)$, while their realisations are written as lowercase letters (e.g., $Z(\omega)=z$, or simply, $Z=z$). If $X$ is a multivariate Gaussian r.v., then its probability law is denoted $\mathcal{N}(\sd x_t; \mu, \Sigma)$, where $\mu$ is the mean and $\Sigma$ is the covariance matrix.
    
\end{itemize}

\color{black}
\section{Background and problem statement}\label{sec2}

\subsection{State space models}

Let $(\Omega,\mF,\mbP)$ be a probability space where $\Omega$  is the sample space, $\mF$  is a $\sigma$ -algebra, and $\mbP$ is a probability measure. On this space, we consider two stochastic processes:
\begin{itemize}
\item  The signal or state process $X = \{ X_t\}_{t\geq 0}$, taking values in a set $\mX  \subseteq  \mbR^{d_x}$.
 \item The observation process $Y = \{ Y_t\}_{t\geq 1}$, taking values in a set $\mY  \subseteq  \mbR^{d_y}$.
 \end{itemize}
 We refer to $\mX$ as the state (or signal) space, while $\mY$ is the observation space.
We assume that the state process evolves over time according to the family of Markov
kernels
$$K_t(x_{t-1}, A) = \mbP (X_t \in A \vert X_{t-1} = x_{t-1}),$$ 
where $A \in  \mB (\mX )$ and $x_{t-1} \in  \mX$. The observation process is described by the conditional
distribution of the observation $Y_t$ given the state $X_t$. Specifically, we assume that $Y_t$ has
a conditional pdf $g_t(y_t\vert x_t)$ w.r.t. a reference measure $\lambda$ (typically, but not necessarily, the
Lebesgue measure), given the state $X_t = x_t$. The observations are assumed to be conditionally
independent given the states.

Throughout the paper we assume arbitrary but fixed observations $\{ Y_t = y_t\}_{t\geq 1}$, and we write $g_t(x_t) := g_t(y_t\vert x_t)$ for conciseness and
to emphasize that $g_t$ is a function of the state $x_t$, i.e., we use $g_t(x)$ as the likelihood of $x \in  \mX$ 
given the observation $y_t$.

The state process $X_t$ with initial probability distribution $\pi_0(\sd x_0)$ and Markov transition kernels $K_t( x_{t-1}, \sd x_t)$ together with observation process $Y_t$, linked to $X_t$ by the pdfs $g_t(y_t\vert x_t)$, form the typical structure of a state space Markov model. Following \cite{crisan2020stable} we refer to the triple $\mM  = (\pi_0, K, g)$, where $K  = \{ K_t\}_{ t\geq 1}$ is the family of Markov kernels for the process $X_t$ and $g = \{ g_t\}_{ t\geq 1}$ is the family of likelihoods generated by the observations $\{ Y_t = y_t\}_{t\geq 1}$, as the SSM. As shown in Section \ref{2.2}, the triple $\mM$  encompasses all the necessary components to define the conditional probability distribution of the state $X_t$ given the observations $Y_{1:t} = y_{1:t}$ or the predictive distribution of $Y_t$ giving the observations $Y_{1:t-1} = y_{1:t-1}$, for every $t \geq  1$. These conditional probability distributions are the main focus of this paper, and therefore we equate $\mM$  with the SSM itself.

\subsection{Bayesian filter}\label{2.2}
 
The filtering problem consists in the computation of the probability law $\pi_t(\sd x):=\mbP (X_t\in d x \vert Y_{1:t}=y_{1:t})$ of the state $X_t$ given a sequence of observations $Y_{1:t}=y_{1:t}$. It is relatively straightforward to use Bayes' rule in order to obtain a relation between $\pi_t$ and the one-step-ahead predictive measure $\xi_t(\sd x)=\mbP (X_t \in \sd x_t \vert Y_{1:t-1}=y_{1:t-1})$ (see for example \cite{crisan2020stable}). Indeed, one can write $\xi_t(\sd x) = \int K_t(x',\sd x)\pi_t(\sd x')$ and for any integrable test function $f:\mX\mapsto\mbR$, it is straightforward to show that
\begin{equation}\label{eq:posteriorpredictive}
\pi_{t}(f)=\frac{\int f(x_t)g_{t}(x_t)K_{t} \pi_{t-1}(\sd x_{t})}{\int g_{y_t}(x_t)K_{t}\pi_{t-1} (\sd x_{t})}=\frac{\xi_{t}(f g_t)}{\xi_{t}(g_{t})}, 
\end{equation}
where we denote $\xi_t(\sd x_t)=K_{t} \pi_{t-1}(\sd x_t)=\int K_t(x_{t-1},\sd x_t)\pi_{t-1}(\sd x_{t-1})$ for conciseness. The normalisation constant $\xi_t(g_t)$ in Eq. (\ref{eq:posteriorpredictive}) is often referred to as the  incremental likelihood at time $t$. It can also be interpreted as the conditional pdf of the observation $Y_t$ given a record of observations $Y_{1:t-1}=y_{1:t-1}$.

\subsection{Model assessment: the Bayesian evidence}

Given a data set $y_{1:T}$, there are different ways to assess ``how good'' a state space model, see \cite{djuric2010assessment}. Possibly the most popular approach is the Bayesian model evidence or marginal likelihood, which can be interpreted as a quantitative indicator of how well a model explains the observed data, while integrating out uncertainties in model parameters and latent states. In particular, a higher Bayesian evidence is usually interpreted as a better fit to the data.

To be specific, the Bayesian evidence of model $\mM=(\pi_0,K,g)$ for data $Y_{1:T}=y_{1:T}$ is denoted $p_T(y_{1:T}\vert \mM)$ and, by a simple marginalisation of the joint distribution of $y_{1:T}$ and $x_{1:T}$, it can be written as  

$$p_T(y_{1:T}\vert \mM)=\int \cdots \int g_t(y_t \vert x_t) \xi_t(\sd x_t)
\cdots g_{1}(y_{1} \vert x_{1}) \xi_{1}(\sd x_{1})=\prod_{t=1}^T \xi_t(g_t),$$
hence, the marginal likelihood at time $T$ is computed as the product of the incremental likelihoods up to time $T$. In most practical applications the quantity of interest is the log-evidence $\log p_T(y_{1:T}\vert \mM)= \sum_{t=1}^T \log \xi_t(g_t),$ which can be more easily computed or approximated.

In problems involving the comparison of two models, $\mM$ and $\mM'$, and a data set $Y_{1:T}=y_{1:T}$, model $\mM$ is considered a better fit than model $\mM'$ if, and only if, $\log p_T(y_{1:T}\vert \mM)\geq \log p_T(y_{1:T}\vert \mM').$


\subsection{Problem statement}

For a given data set  $Y_{1:T}=y_{1:T}$ and a given state space model $\mM=(\pi_0,K,g)$, we seek a methodology to modify $\mM$ in a systematic way that yields an ``improved" model, denoted $\mM^{\alpha}$ with a higher Bayesian evidence, i.e., $\log p_T(y_{1:T}\vert \mM^\alpha)\geq \log p_T(y_{1:T}\vert \mM)$.


Our approach towards increasing the evidence of the base model $\mM$ consists in adapting the Markov kernels $K_t$ to the observed data $y_{1:T}$. Note that the Markov kernels govern the dynamics of the state process.

To be specific, we are interested in a sequential procedure that, at time $t$, takes the new observation $Y_t=y_t$ and uses it to convert the original kernel $K_t$ into an updated one $K_t^{\alpha}$. As a result, we sequentially construct a new model $\mM^\alpha=(\pi_0,K^\alpha,g)$, incorporating the adjusted kernels $K^\alpha=\{K_t^\alpha\}_{t\geq 1}$.

Ideally, the methodology should ``refine" the initial model $\mM$, in the sense of increasing the Bayesian evidence with slight changes to the dynamics. This modification is data-driven and carried out in a systematic, automatic manner that can be implemented easily for a broad class of models.

\section{Nudging schemes}\label{sNudging}

\subsection{Nudging}

In this paper we intend to adaptively modify the transition kernel $K_t(x',\sd x)$
to better align with the data, resulting in an improved model. Given observed data $Y_t=y_{t}$, at each time step $t$ we adjust the Markov kernel $K_t(x',\sd x)$ to obtain the modified kernel 
\beq\label{eq:NudgingKernel}
K_t^\alpha(x_{t-1},\sd x_t) := \int \delta_{\alpha_t(x_t')}(\sd x_t) K_t(x_{t-1},\sd x_t'),
\eeq
where $\delta_{\alpha_t(x_t')}$  denotes the Dirac delta measure centred at $\alpha_t(x_t')$, and $\alpha_t:\mX \to \mX$ is a transformation of the state space into itself that depends on the observation $Y_t=y_t$.  By construction, the map $\alpha_t$ increases the value of the function $g_t$, i.e. $g_t(x)\leq g_t(\alpha_t(x)),$ for any $x\in \mX$. 
The modified kernel in Eq. (\ref{eq:NudgingKernel}) yields a new model $\mM^{\alpha}=\{\pi_0,K^{\alpha},g\}$, where $K^{\alpha}=\{K_t^{\alpha}: t\geq 1\}$, for which the Bayesian evidence can be computed as



\beq\label{OBNE}
p_T(y_{1:T}\mid \mM^{\alpha}) = \prod_{t=1}^T  \xi_t^{\alpha}(g_t) 
\nn
\eeq
and the predictive measure are recursively computed as $\xi_t^{\alpha}=K_t^{\alpha}\pi_{t-1}^{\alpha}$. The posterior marginals are
\beq
\pi_t^\alpha(f) = \frac{\xi_t^\alpha(fg_t)}{\xi_t^\alpha(g_t)}, \quad \text{ for } t=1,2,...,
\label{eq:Nudgingpostpre}
\eeq
and the prior is the same as in the original model. 

\begin{remark}
We have introduced a new model $\mM^\alpha$ derived from modifications made to the transition kernel. These changes modify the system behaviour, potentially differing in its dynamics compared to the original system $\mM$. This can significantly impact the system evolution and must be carefully considered in the analysis.
\end{remark}

\subsection{Parametric nudging scheme}

The key element of a nudging scheme is the sequence of maps $\alpha_t,\; t\geq 1$. 

\begin{definition}\label{def:NudParT}
The set of maps $\{\alpha_t(x,\gamma):\mX \times \mbR^+\to \mX,$ $t\in \mbN\}$ is a family of parametric nudging transformations if and only if it satisfies the conditions below: 

\begin{enumerate}[i)] 
\item The transformation $\alpha_t$ is continuous in $\gamma$ and
\begin{equation}\label{A0}
\lim_{\gamma\rw 0} \alpha_t(x,\gamma) = x, \quad \forall x\in \mX, ~~\forall t \ge 1. 
\end{equation}
\item There are 
intervals $[0,\Gamma_t)$ with $\Gamma_t>0$, such that   
\begin{equation}\label{eq:IncreasingNudging}
g_t(\alpha_t(x,\gamma))-g_t(x)\geq 0 , \; \forall (x,\gamma) \in \mX \times [0,\Gamma_t).
\end{equation} 
    \item For every $t\geq 1$ and $\gamma\in(0,\Gamma_t)$
    \begin{equation}\label{A1}
\Delta_{g_t}(\gamma):=
\int_{\mX} \left[g_t(\alpha_t(x,\gamma))- g_t(x)\right]  \xi_t(\sd x)>0.
    \end{equation}
\end{enumerate}
\end{definition}
Hereafter we limit our discussion to nudging parametric transformations. Although other possibilities exist, a natural choice for the map $\alpha_t$ is to construct it as a single step of a gradient-ascent algorithm for the maximisation of $\log g_t$, as specifically described in Section \ref{SGAN}. 

Together with Definition \ref{def:NudParT} for parametric nudging transformations, we also assume some mild regularity of the model $\mM=(\pi_0,K,g)$.

\begin{assumption}\label{A2}
The model $\mM=(\pi_0,K,g)$ satisfies the conditions below.
\begin{enumerate}[i)] 
    \item The likelihood functions $g_t(x)$ are continuous and bounded, i.e. $\norm{g_t}_{\infty}< \infty,\; t\geq 1.$ 
    \item The transition kernels are continuous with respect to the total variation norm, i.e., for every $x\in \mX$ and $\epsilon>0$, there exists $\delta_{\epsilon,x,t}>0$ such that
$$\norm{K_t(x,\cdot)-K_t( x',\cdot)}_{TV}\leq \epsilon, \;\; \forall x'\in \mX, \text{ whenever } \norm{x-x'}\leq \delta_{\epsilon,x,t}. $$
\end{enumerate}
\end{assumption} 

\begin{remark}
In Section \ref{NKalman}, we show that the transition kernels for linear-Gaussian SSMs are continuous with respect to the total variation norm. 
\end{remark}

We are now ready to state our first result on the ``improvement" of the nudged model $\mM^\alpha$ over the original model $\mM.$

\begin{theorem}\label{ISSL}
   Let $\{\alpha_t\}_{t \in \mathbb{N}}$ be a family of nudging parametric transformations as in Definition \ref{def:NudParT}. If Assumption  \ref{A2} holds, then there exists a sequence of positive parameters $\gamma_{1:T},$ such that 
   $$
   p_T (y_{1:T} \vert \mM ^{\alpha}) \ge p_T(y_{1:T}\vert \mM),
   $$ i.e., model $\mM^{\alpha}$ has a higher Bayesian evidence than model $\mM$.
\end{theorem}

    In Appendix \ref{App1} we introduce an alternative nuging scheme that relies on the same maps $\alpha_t, \; t=1,..,T,$ and generates a closely related (but different) model $\bar \mM^\alpha$. This new model yields the same Bayesian evidence as $\mM^\alpha$, i.e. $p_T(y_{1:T} \vert \mM ^{\alpha})=p_T(y_{1:T} \vert \bar \mM ^{\alpha})$ but it is easier to analyse. Then, using $\bar \mM^\alpha$ we prove Theorem \ref{ISSL} by an induction argument in Appendix \ref{App2}.

\begin{remark}
Ensuring that the Bayesian evidence is increased, $p_T (y_{1:T} \vert \mM ^{\alpha}) \ge p_T(y_{1:T}\vert \mM)$, via the proposed nudging methodology requires a certain balance between the increments $(g_t\circ\alpha_t)(x)-g_t(x)>0$ in the likelihoods and the preservation of the original dynamics, i.e., keeping the nudged kernels $K_t^\alpha$ `close' to the original $K_t$ in total variation distance. Theorem \ref{ISSL} provides a theoretical guarantee that this balance can always be attained within the framework of the parametric nudging transformations in Definition \ref{def:NudParT}, and the class of models that satisfy Assumption \ref{A2}. It is, however, an existence result that does not provide an explicit procedure to identify suitable parameters $\gamma_{0:T}$ given the model $\mM$ and the data set $y_{1:T}$. Specific families of models as well as a systematic way of constructing the parametric maps $\alpha_t$ from the data $y_{1:T}$ are investigated in the remaining of this section.
\end{remark}

\subsection{Parametric models} \label{ssParametricModels}

Assume a model $\mM_{\theta}=(\pi_{0,\theta},K_\theta,g_\theta)$, where the prior $\pi_{0,\theta}$, the kernels $K_\theta=\{K_{t,\theta}\}_{t\geq 1}$ and the likelihood functions $g_t=\{g_{t,\theta}\}_{t\geq 1}$ are indexed by a parameter vector $\theta$. Given a data set $Y_{1:T}=y_{1:T}$, the model marginal likelihood is $$p_T(y_{1:T}\vert \mM_\theta)=\prod_{t=1}^T \xi_{t,\theta}(g_{t,\theta}),$$ where the (parametrised) predictive measures $\xi_{t,\theta}$ are computed in the usual way. One common form of model mismatch occurs when the choice of $\theta$ does not accurately reflect the dynamics of the real-world system.

In order to fit $\mM_\theta$ to the observed data, a standard approach is to compute the maximum likelihood estimator (MLE) of the parameters vector $\theta$, i.e, we obtain 
$$
\theta^\star = \arg\max_{\theta}\; p_T(y_{1:T}\vert \mM_{\theta}).
$$
For many problems, the MLE $\theta^\star$ may not be easy to compute (it may be intractable). In practice, it may only be possible to compute a suboptimal estimator $\tilde\theta$ such that $p_T(y_{1:T}|\mM_{\tilde \theta}) < p_T(y_{1:T}|\mM_{\theta^\star})$. It is a natural question to ask whether, in this setting, a nudging scheme can be used to ``bridge the gap'' (at least partially) between the marginal likelihoods $p_T(y_{1:T}|\mM_{\tilde \theta})$ and $p_T(y_{1:T}|\mM_{\theta^\star})$. To be specific, Theorem \ref{ISSL} says that, under regularity assumptions, it is possible to find a parametric nudging transformation such that $p_T(y_{1:T}|\mM_{\tilde\theta}^\alpha) \ge p_T(y_{1:T}|\mM_{\tilde\theta})$ and the problem is to establish some guarantee that $p_T(y_{1:T}|\mM_{\tilde\theta}^\alpha)$ is ``reasonably close'' to $p_T(y_{1:T}|\mM_{\theta^\star})$. In this section, we address precisely this issue. 

For our analysis we assume that the model is Lipschitz in each of its components. In particular, if we let $\Theta$ denote the parameter space, then we make the following assumption.


\begin{assumption}\label{A3}
There are finite constants $C_0$ and $\{G_t,\kappa_t\}_{t\geq 1}$ such that, for any $\theta,\theta'\in \Theta$
\begin{enumerate}[i)]
    \item  
$\norm{\pi_{0,\theta}-\pi_{0,\theta'}}_{TV}\leq C_0 \norm{\theta-\theta'},$ 
\item  $\norm{g_{t,\theta}-g_{t,\theta'}}_{\infty} \leq G_t \norm{\theta-\theta'},$
\item $\norm{K_{t,\theta}(x,\cdot)-K_{t,\theta'}(x,\cdot)}_{TV}\leq \kappa_t 
\norm{\theta-\theta'}.$ 
\end{enumerate}
\end{assumption}

It is straightforward to show that Assumption \ref{A3} implies that the marginal likelihood is Lipschitz itself (see Appendix \ref{App3}), i.e., there exists a finite constant $L_T$ such that
\beq\label{eq:Lipschitzparametriclikelihood}
\abs{p_T(y_{1:T}\vert \mM_\theta)- p_T(y_{1:T}\vert \mM_{\theta'})}\leq L_T\norm{\theta-\theta'}\quad \text{for any } \theta,\theta'\in \Theta.
\eeq

Let $p_T(y_{1:T}\vert \mM^{\alpha}_\theta)$ and $p_T(y_{1:T}\vert \mM_{\theta^\star})$ denote the Bayesian evidence of the  model with nudging for an arbitrary parameter $\theta$ and the original model for the MLE $\theta^\star$, respectively. According to Theorem \ref{ISSL}, there exists a sequence of parameters $\gamma_{0:T}^\theta$ for which 
$$
\Delta_T^\alpha(\theta):=p_T(y_{1:T}\vert \mM_{\theta}^{\alpha})- p_T(y_{1:T}\vert \mM_{\theta})\geq 0
$$
and we refer to $\Delta_T^\alpha(\theta)$ as the \textit{nudging gain}. The main result of this section follows.

\begin{corollary}\label{cor:NudgingMLE}
 Let $\{\alpha_t\}_{t \in \mathbb{N}}$ be a family of parametric  nudging transformations. 
    If Assumptions \ref{A2} and \ref{A3} hold then there exists $\gamma_{0:T}^\theta$ such that
    \beq \label{thm:NudgingMLE}
    p_T(y_{1:T}\vert \mM_{\theta}^\alpha) \in \left[p_T(y_{1:T}\vert \mM_{\theta^\star})-L_T\norm{\theta^\star-\theta},\; p_T(y_{1:T}\vert \mM_{\theta^\star})+ \Delta_T^{\alpha}(\theta) \right],
    \eeq
where $\Delta_T^\alpha(\theta) \ge 0$, for any $\theta \in \Theta$.

\end{corollary}

\begin{proof} Using inequality \eqref{eq:Lipschitzparametriclikelihood} for the MLE $\theta^\star$, we readily obtain 
\begin{equation}
    0\leq p_T(y_{1:T}\vert \mM_{\theta^\star})-p_T(y_{1:T}\vert \mM_{\theta})\leq L_T \norm{\theta^\star-\theta}, \; \forall \theta\in \Theta.
\label{eq:thm32-1}    
\end{equation}
On the other hand, for any $\theta\in \Theta$, Theorem \ref{ISSL} implies that there is a sequence $\gamma_{1:T}^\theta$ such that
\begin{equation*}
    0\leq p_T(y_{1:T}\vert \mM^{\alpha}_{\theta})-p_T(y_{1:T}\vert \mM_{\theta})= \Delta_T^{\alpha}(\theta)
\end{equation*}
and we can easily use the expression above to rewrite the difference $p_T^{\alpha}(y_{1:T}\vert \mM_{\theta}^{\alpha})- p_T(y_{1:T}\vert \mM_{\theta^\star})$ as 
\begin{equation}
p_T(y_{1:T}\vert \mM_{\theta}^{\alpha})- p_T(y_{1:T}\vert \mM_{\theta^\star}) = \Delta_T^{\alpha}(\theta) + p_T(y_{1:T}\vert \mM_{\theta})-p_T(y_{1:T}\vert \mM_{\theta^\star}).
\label{eq:thm32-2}
\end{equation}
Finally, combining inequality \eqref{eq:thm32-1} with Eq. \eqref{eq:thm32-2} above we arrive at
\begin{equation}\label{eq:MLE_bound}
 - L_T \norm{\theta^\star-\theta}\leq p_T(y_{1:T}\vert \mM_{\theta}^{\alpha})- p_T(y_{1:T}\vert \mM_{\theta^\star})\leq \Delta_T^{\alpha}(\theta),
    \end{equation}
which is equivalent to \eqref{thm:NudgingMLE}. 
\end{proof}

Corollary \ref{cor:NudgingMLE} shows that, when the nudging parameters $\gamma_{1:T}^\theta$ are suitably chosen to ensure $\Delta_T^\alpha(\theta)\ge 0$ (which is always possible by Theorem 3.1), the Bayesian evidence of the nudged model, $p_T(y_{1:T}\vert \mM_{\theta}^{\alpha})$, lies in a neighbourhood of the Bayesian evidence attained with the MLE $\theta^\star$ and the original model, $p_T(y_{1:T}\vert \mM_{\theta^\star})$. More specifically, let us note that:
\begin{itemize}
\item From the left hand inequality in expression \eqref{thm:NudgingMLE}, we see that
$$
p_T(y_{1:T}\vert \mathcal{M}_{\theta^\star}) \le p_T(y_{1:T}\vert \mathcal{M}_{\theta}^\alpha) + L_T \|\theta^\star-\theta\|
$$
which shows that the nudged model $\mM_\theta^\alpha$ attains a Bayesian evidence which is close to the evidence attained with the MLE $\theta^\star$.
\item Since we can choose $\gamma_{1:T}^\alpha$ to ensure $\Delta_T^\alpha(\theta) \ge 0$, the right hand side inequality in expression \eqref{thm:NudgingMLE} shows that it is possible to have $p_T(y_{1:T}| \mathcal{M}_{\theta}^\alpha) \ge p_T(y_{1:T}| \mathcal{M}_{\theta^\star})$ (typically, when $\| \theta - \theta^\star \|$ is small enough).
\end{itemize}

\begin{remark}\label{rem:path_filter_error}
Consider a bounded test function $\varphi: \mathcal{X}^{\otimes T} \to \mathbb{R}$ and denote the path measure of $X_{1:T}$ conditioned on $y_{1:T}$ generated by model $\mM_{\theta^\star}$ as $\Pi_T^{\theta^\star}(\sd x_{1:T})$. Similarly, we denote the path measure of the nudged model $\mM_{\theta}^\alpha$ as $\Pi_T^{\theta,\alpha}(\sd x_{1:T})$. For any bounded test function, a simple calculation (see Appendix~\ref{app:proof_filters}) shows that 
\begin{align}\label{eq:path_error_bound}
\left| \Pi_T^{\theta^\star}(\varphi) - \Pi_T^{\theta,\alpha}(\varphi) \right| &\leq 2 \|\varphi\|_\infty \frac{|p_T(y_{1:T} | \mathcal{M}_{\theta^\star}) - p_T(y_{1:T} | \mathcal{M}_\theta^\alpha)|}{p_T(y_{1:T} | \mathcal{M}_{\theta^\star})}.
\end{align}
Therefore, we can attain a minimum error when $|p_T(y_{1:T} | \mathcal{M}_{\theta^\star}) - p_T(y_{1:T} | \mathcal{M}_\theta^\alpha)| \to 0$, that is, if $\{ \alpha_t(\cdot,\gamma_t) \}_{1 \le t \le T}$ is chosen such that this quantity is minimised. Hence, it is important to design the nudging transformations $\{\alpha_t(\cdot,\gamma_t)\}_{1 \le t \le T}$ carefully to avoid overshooting.
\end{remark}

The remark above relies implicitly on the assumption that the chosen class of models $\mathcal{M}_\theta$ is a good fit to the sequence of observations $y_{1:T}$ as a parametric family, hence $\Pi_T^{\theta^\star}(\varphi)$ is a desirable estimator of $\varphi(X_{1:T})$. If the chosen statistical family $\{\mathcal{M}_{\theta} : \theta \in \Theta\}$ does not contain a single desirable statistical model, the above discussion may be different and a nudged kernel with higher likelihoods may still attain more desirable results.

\subsection{Gradient ascent nudging transformation}\label{SGAN}

While nudging can be implemented in several ways \cite{akyildiz2019sequential}, a natural approach is to use the gradient of $\log g_t$ to shift the Markov kernel $K_t$ towards regions of the state space  $\mX$ where the likelihood is higher.

To be specific, let Assumption \ref{A2} hold and, additionally, assume that the functions $\log g_t(x), \; t=1,...,T$, are sufficiently differentiable. We construct a nudging map $\alpha_t:\mX \times [0,\Gamma_t]\to \mX$ of the form 
\beq\label{eq:NudgingGradientAscent}
\alpha_t(x,\gamma):=x+\gamma \nabla \log g_t(x),
\eeq
where $\nabla = \left[ \frac{\partial}{\partial x_1}, \ldots, \frac{\partial}{\partial x_{d_x}} \right]^\mathsf{T}$ is the gradient operator.

An obvious question is whether \eqref{eq:NudgingGradientAscent} is compatible with Definition \ref{def:NudParT}. Clearly, $\alpha_t(x,\gamma)$ is continuous over $\gamma$ and $\lim_{\gamma\to 0}\alpha_t(x,\gamma)\to x$, for all $x\in\mX$, whenever $\norm{\nabla \log g_t(x)}< \infty$, hence Eq. \eqref{A0} holds. As for Eq. \eqref{eq:IncreasingNudging}, it is satisfied when $\nabla \log g_t(x)$ is $L_t$-Lipschitz continuous, i.e., when there is a sequence of constants $L_t < \infty$ such that
\beq\label{eq:LipschitzGradientlog}
\norm{\nabla \log g_t(x)-\nabla \log g_t(x')} \leq L_t \norm{x-x'}, \; \; \forall x,x' \in \mathbb{R}^{d_x}, \; t\in \mbN.
\eeq
We resort to the proposition below

\begin{proposition} \label{IGA}
Assume that the function $f:\mX\to \mbR$ is differentiable and its gradient is $L$-Lipschitz continuous. Then for all $x\in\mX$ such that  $\nabla f(x)\neq 0$, we have
   \beq
   f(x+\gamma \nabla f(x))\geq f(x)+\gamma \left(1-\frac{\gamma L}{2} \right)\norm{\nabla f(x)}^2> f(x), \; \; \forall \gamma\in(0,2/L ).
   \eeq
\end{proposition}
This is just a slight variation of Theorem 3 in \cite{polyak1963gradient}. See  Lemma 1.2.3 and Eq. (1.2.12) of \cite{nesterov2013introductory} for an explicit proof. 
If we apply Proposition \ref{IGA} to the log likelihood functions $\log g_t$, we obtain
\beq\label{eq:Increaselogg}
\log g_t(\alpha_t(x,\gamma))\geq \log g_t(x)+\gamma \left(1-\frac{\gamma L_t}{2} \right)\norm{\nabla \log g_t(x)}^2 \ge \log g_t(x),
\eeq
with equality only if $\gamma=0$ or $\nabla \log g_t=0$. Taking exponentials on the three terms of \eqref{eq:Increaselogg} yields
\beq \label{eq:ExponentialIncreaseg}
g_t(\alpha(x,\gamma))  \geq e^{\gamma \left(1-\frac{\gamma L_t}{2} \right)\norm{\nabla \log g_t(x)}^2}g_t(x) \ge g_t(x), \quad \forall x\in \mX, \gamma\in[0,2/L_t).
\eeq
Hence, if there is a set $A_t \subseteq \mX$ such that $\nabla \log g_t(x) \ne 0$ for all $x \in A_t$ and $\xi_t(A_t)>0$, it follows from \eqref{eq:ExponentialIncreaseg} that  
\begin{equation}\label{eq:Deltag}
\Delta_{g_t}(\gamma)=\int_{\mX} [g_t(\alpha_t(x,\gamma))-g_t(x)] \xi_t(\sd x)>0, \;\; \text{ for all } \; \gamma\in (0,2/L_t).
\end{equation}

The inequalities \eqref{eq:ExponentialIncreaseg} and \eqref{eq:Deltag}
imply that Eq. \eqref{eq:IncreasingNudging} and Eq. \eqref{A1} are satisfied, and we state the following corollary of Theorem \ref{ISSL} for the special case when nudging is implemented as a gradient step. 

\begin{corollary}\label{cor:GANS}
For $t=1,..,T$, let $\alpha_t$ have the form in \eqref{eq:NudgingGradientAscent} and let $Y_{1:T}=y_{1:T}$ be an arbitrary but fixed data set. If 
\begin{enumerate}[(a)]
\item\label{ass_a} $\nabla \log g_t(x)$ is $L_t$-Lipschitz continuous,
\item there are sets $A_t \subseteq \mX$ such that $\xi_t(A_t)>0$  and $\nabla \log g_t(x)\neq0$ for all $x \in A_t$, and
\item Assumption \ref{A2}. ii) holds,
\end{enumerate}
then there exists a positive sequence $\gamma_{0:T}$ (depending on $\mM$ and $y_{1:T}$) such that
\beq
p_T(y_{1:T}|\mM^\alpha) \ge p_T(y_{1:T}|\mM).
\label{eqCor}
\eeq
\end{corollary}

\begin{proof}
The gradients $\nabla \log g_t$ are finite and the sets $A_t\subseteq \mX$ exist by assumption, and we have seen that this is sufficient for Eq. \eqref{A0}, Eq. \eqref{eq:IncreasingNudging} and Eq. \eqref{A1} to hold. Hence, $\{\alpha_t\}_{t\in\mbN}$ is a family of parametric nudging transformations as described by Definition \ref{def:NudParT}. Since Assumption \ref{A2} on the model $\mM$ is given, the inequality \eqref{eqCor} follows directly from Theorem \ref{ISSL}.
\end{proof}

\begin{remark}
Assumption (\ref{ass_a}) in the statement of Corollary \ref{cor:GANS} may be satisfied or not depending on the likelihood model $g_t$. A simple example where the assumption does not hold corresponds to the multiplicative-noise observation model
\beq
Y_t=\exp\left\{ \frac{X_t}{2} \right\} Z_t, \quad Z_t \sim \mN(0,1),
\label{eqMultiplicative}
\eeq
which is commonly found in stochastic volatility models \cite{tsay2005analysis}. It is clear that $Y_t|X_t \sim \mN(0,\exp\{X_t\})$, hence, for given $Y_t=y_t$, the likelihood function becomes $g_t(x) \propto \exp\left\{-\frac{1}{2}\left(x+\exp\{-x\}y_t^2\right) \right\}$ and it is straightforward to show that $\nabla \log g_t(x)$ is not Lipschitz.
\\
The main difficulty with \eqref{eqMultiplicative} arises from the noise being multiplicative. For example, if we assume the nonlinear observation in additive Gaussian noise 
$$
Y_t = f(X_t) + Z_t, \quad Z_t \sim \mN(0,C),
$$
where $f:\mX \mapsto \mY$ is some possibly nonlinear map, then, for given $Y_t=y_t$, the likelihood function is
$$
g_t(x) \propto \exp\left\{ -\frac{1}{2}(y_t-f(x))^\top C^{-1} (y_t-f(x)) \right\}
$$
and 
$$
\nabla \log g_t(x) = -C^{-1}(y-f(x)\nabla f(x)),
$$
where $\nabla f(x)$ is the Jacobian matrix of $f(x)$. Assumption (\ref{ass_a}) holds when the function $f(x)\nabla f(x)$ is Lipschitz, which can be easily tested in most cases. Linear observations with additive Gaussian, Student-t or Cauchy noise can easily be shown to satisfy Assumption (\ref{ass_a}). 
\\
Finally, note that the Lipschitz continuity of $\nabla \log g_t(x)$ is a sufficient condition for Eq. \eqref{eq:IncreasingNudging} to hold, but it is possibly non-necessary. 
\end{remark}

\begin{remark}\label{rem:maximiser_mle_nudging}
Optimisation-based implementations of nudging should be done carefully in light of Remark~\ref{rem:path_filter_error} under parameter misspecification. In particular, we propose the nudging transformation $\alpha_t(x, \gamma_t)$ in \eqref{eq:NudgingGradientAscent} as a gradient step with step-size $\gamma_t$, but without careful implementation, the overshooting problem mentioned in Remark~\ref{rem:path_filter_error} can be problematic. For example, consider a map $\alpha_t(x, \gamma_t)$ for a given likelihood $g_t(x)$ that returns $x_t^\star \in \arg \max_{x\in\mathcal{X}} g_t(x)$, that is, the maximiser of the likelihood (as for log-concave likelihoods, this would eventually happen if nudging were run for many steps). It can be easily shown that this results in a strictly positive difference for the marginal likelihoods in \eqref{eq:path_error_bound} for any $\theta$. In particular, let $\mM^{\alpha^\star}=(\pi_0,K^{\alpha^\star},g)$ denote the nudged model where the nudged kernel is degenerate, i.e. $K_t^{\alpha^\star}(x_{t-1},\sd x_t)=\delta_{x_t^\star}(\sd x_t)$. Note that in this case, the filter is independent of any transition kernel parameter $\theta$. Then 
$$
p_T(y_{1:T}\vert \mM_{\theta^\star}) = \int \prod_{t=1}^T g_t(x_t) K_{t,\theta^\star}(\sd x_t | x_{t-1}) \pi_0(\sd x_0) < \prod_{t=1}^T g_t(x_t^\star) = p_T(y_{1:T}\vert \mM^{\alpha^\star}),
$$
i.e., $|p_T(y_{1:T}\vert \mM^{\alpha^\star}) - p_T(y_{1:T}\vert \mM_{\theta^\star}) | > 0$. This results in higher estimation errors compared to a less aggressive nudging map $\alpha_t$ which can satisfy $|p_T(y_{1:T}\vert \mM^{\alpha}) - p_T(y_{1:T}\vert \mM_{\theta^\star}) | \approx 0$ (see Remark~\ref{rem:path_filter_error}). This shows that one should not blindly maximise this likelihood but instead choose an empirically well performing step size $\gamma$.
\end{remark}

\begin{remark}
Throughout this Section \ref{SGAN} we have assumed that the nudging transformation maps the state space $\mX$ onto itself. However, the transformation defined in Eq.~\eqref{eq:NudgingGradientAscent} does not necessarily satisfy this property when $\mX$ is bounded. In such case, we can define a parametric family of nudging transformations that take projected gradient ascent steps, instead of standard gradient ascent steps, and still yields the same result in Corollary \ref{cor:GANS}, provided that $\mX$ is closed and convex. This is discussed in detail in Appendix~\ref{A:PGAN}. Furthermore, the analysis in Appendix~\ref{A:PGAN} can be extended mutatis mutandis to proximal gradient algorithms \citep{garrigos2023handbook}, which allow to handle nondifferentiable likelihoods.
\end{remark}

\subsection{Linear and Gaussian models}\label{NKalman}

In this section we explore the application of the nudging methodology to linear Gaussian systems. In particular, we consider the model $\mM=\{\pi_0,K,g\}$ where $\pi_0(\sd x)=\mathcal{N}(\sd x; m_0, P_0)$, i.e., $\pi_0$ is a Gaussian law with mean $m_0$ and covariance matrix $P_0$. The Markov kernels $K_t$ and the likelihood functions $g_t$ are also Gaussian, i.e.,  
\beq \label{eq: LinearGausian_K}
K_t(x_{t-1},\sd x_t)=\mathcal{N}(\sd x_t; A_t  x_{t-1}, Q_t)\eeq and 
\beq \label{eq: LinearGaussian_g}
g_t(x_t)\propto \exp\{ -\frac{1}{2} (y_t-C_tx_t)^\top R_t^{-1}(y_t-C_tx_t)\},
\eeq
respectively. We assume that the model parameters $A_t,Q_t,C_t$ and $R_t$ are known (for every time $t=1,...,T$).

We apply the gradient-ascent nudging scheme of Section \ref{SGAN} to the model $\mM$ described above. In particular the nudging map $\alpha_t$ of Eq. \eqref{eq:NudgingGradientAscent} becomes
\beqa
\alpha_t(x,\gamma_t) &=& x + \gamma_t \nabla \log g_t(x) \nn \\
&=& x + \gamma_t C_t^\top R_{t}^{-1}(y_t - C_t x) \nn \\
&=& (I - \gamma_t C_t^\top R_t^{-1} C_t) x +  \gamma_t C_t^\top R_t^{-1} y_t,
\label{eq:NudgingmapGaussian}
\eeqa
where the second equality comes from the (straightforward) calculation of $\nabla \log g_t(x)$ and the third equality is obtained by re-arranging terms. Let us note that
\beq \label{eq: LipschitzConstant_LinearGaussian}
\norm{\nabla \log g_t(x)-\nabla \log g_t(x')}\leq \norm{C_t^\top R_t^{-1}C_t}\norm{x-x},
\eeq 
which implies that, in this case, the Lipschitz constant is given by $L_t=\norm{C_t^\top R_t^{-1}C_t}$, and we need to select $\gamma \in (0,2/L_t)$, at each time step, in accordance with Eq. (\ref{eq:Deltag}). Moreover, it can be seen that for $\gamma \in [0,1/L_t)$ the inverse of the nudging transformation $\alpha^{-1}_t(x)$ exists and is given by 
\beq \label{eq:NudgingInverse}
\alpha^{-1}_t(x)=(I - \gamma_t C_t^\top R_t^{-1} C_t)^{-1}(x-\gamma_t C_t^\top R_t^{-1} y_t).
\eeq
Finally, it can be seen from Eq. \eqref{eq:NudgingmapGaussian} that the resulting nudging is an affine map of the state, which allows us to derive the modified kernel $K_t^{\alpha}(x_{t-1}, \sd x_t)$ in closed form. To be specific, one readily obtains 
\begin{align}
K_t^\alpha(x_{t-1}, \sd x_t) &=
 \mathcal{N}(\sd x_t; M_t A_t x_{t-1} + \gamma_t C_t^\top R^{-1}_t y_t, M_t Q_t M_t^\top) \label{eq:nudged_Gaussian_kernel},
\end{align}
where
\begin{align*}
M_t = I - \gamma_t C_t^\top R_t^{-1} C_t.
\end{align*}

The nudged model $\mM^\alpha=\{\pi_0,K^\alpha,g\}$ is affine and Gaussian, which implies that the predictive and filtering laws, $\xi_t^\alpha$ and $\pi_t^\alpha$, respectively, can be computed exactly using a KF. To be specific, we have $\xi_t^\alpha(\sd x) = \mathcal{N}(\sd x; \tilde \mu_t,\tilde P_t)$ and $\pi_t^\alpha(\sd x)=\mathcal{N}(\sd x; \mu_t, P_t)$ where the posterior means ($\tilde \mu_t,\mu_t$) and covariances ($\tilde P_t, P_t$) are computed recursively as 
\beq
\left\{ \begin{array}{lcc} \tilde P_t = M_tA_t P_{t-1}A_t^\top M_t^\top +M_tQ_tM_t^\top,\\
\tilde \mu_t = M_tA_t \mu_{t-1} + \gamma_t C_t^\top R_t^{-1}y_t. 
\\ \end{array} \right.
\eeq
\beq
\hspace{-.5cm}\left\{ \begin{array}{lcc} 
S_t=C_t\tilde P_t C_t^\top+R_t,
\\ 
\mu_t = \tilde \mu_t +\tilde P_tC_t^\top S_t^{-1}(y_t-C_t\tilde \mu_t),
\\ 
P_t =\tilde P_t-\tilde P_tC_t^\top S_t^{-1}C_t\tilde P_t.
\end{array} \right.
\eeq
Similar algorithms, with additive correction terms (obtained by different arguments), have been investigated, especially in continuous-time settings (see, e.g., \cite{pathiraja2024connections}).

Even if the laws $\xi_t^\alpha$ and $\pi_t^\alpha$ can be obtained exactly, the question remains whether there is a sequence $\gamma_{1:T}$ such that the marginal likelihood is improved by nudging, i.e., whether $p_T(y_{1:T}\vert \mM^\alpha)\geq p_T(y_{1:T}\vert \mM)$. To answer this question we examine whether model $\mM$ satisfies the assumptions of Corollary \ref{cor:GANS}. 
Since the gradient of $\log g_t$ has the form
$$
\nabla \log g_t(x)=C_t^\top R_t^{-1}(y_t-C_tx_t),
$$ 
it follows that $\norm{\nabla \log g_t(x)} < \infty$, for all $x \in \mbR^{d_x}$ and $t = 1, \ldots, T$. Furthermore, the probability law $\xi_t(\sd x)$ is Gaussian (for every $t$), hence for any cell $I_t \subseteq \mbR^{d_x}$ with positive Lebesgue measure we have $\xi_t(I_t)>0$, $t=1,..,T$. Finally, the likelihoods $g_t$ are continuous and bounded, which accounts for Assumption \ref{A2}.i), 
hence it only remains to prove that Assumption \ref{A2}.ii) holds for the linear and Gaussian model $\mM$. 

We proceed using Proposition 2.1 in \cite{devroye2018total}: if $\Sigma_1$ and $\Sigma_2$ are positive definite covariance matrices, then
\begin{equation}
\begin{split}
 \hspace{-.25cm} \vert \vert \mathcal{N}(\sd x; \mu_1, & \Sigma_1)-\mathcal{N}(\sd x; \mu_2,\Sigma_2)\vert \vert_{TV} \leq \\
&\frac{1}{2} \sqrt{\mathrm{Tr}(\Sigma_1^{-1}\Sigma_2-I)+(\mu_1-\mu_2)^\top\Sigma_1^{-1}(\mu_1-\mu_2)-\log(\det(\Sigma_2\Sigma_1^{-1}))}.
\end{split}
\label{eq:DMR}
\end{equation}
In our case, 
$$
\norm{K_t(x_{t-1},\sd x_t)-K_t(x_{t-1}',\sd x_t)}_{TV}=
\norm{\mathcal{N}(\sd x_t; A_t\bar x_{t-1}, Q_t)-\mathcal{N}(\sd x_t; A_t\bar x_{t-1}', Q_t)}_{TV},
$$ 
i.e., comparing to \eqref{eq:DMR} we have $\Sigma_1=Q_t=\Sigma_2$ and $\mu_1=A_t\bar x_t, \mu_2=A_t\bar x_t'$ and the inequality \eqref{eq:DMR} readily implies 

\begin{equation}
\begin{split}
\vert \vert \mathcal{N}(\sd x_t; A_t\bar x_{t-1},Q_t)-  \mathcal{N} (\sd x_t; A_t & \bar x_{t-1}', Q_t)\vert \vert_{TV}  \leq \\
 & \frac{1}{2}  \sqrt{(A_t(\bar x_{t-1}  -\bar x_{t-1}'))^\top  Q_t^{-1}A_t(\bar x_{t-1}-\bar x_{t-1}')}.
\end{split}
\label{eq:TVNormNormal}
\end{equation}
Since $Q_t$ is a positive definite symmetric matrix, its eigenvalue decomposition yields 
\beq
Q_t=U_t^\top\Lambda_t U_t,
\label{eq:SpectralDecQ}
\eeq
where $U_t$ is a unitary matrix and $\Lambda_t$ is a diagonal matrix with the (real and positive) eigenvalues of the matrix $Q_t$. Substituting \eqref{eq:SpectralDecQ} into \eqref{eq:TVNormNormal} yields
\begin{equation}
\norm{\mathcal{N}(\sd x_t; A_t\bar x_{t-1},Q_t)-\mathcal{N}(\sd x_t; A_t\bar x_{t-1}',Q_t)}_{TV}\leq 
\frac{1}{2}\norm{A_t\Lambda_t^{-\frac{1}{2}}U_t}\norm{\bar x_{t-1}-\bar x_{t-1}'}. 
\nn
\end{equation}
Therefore, the linear and Gaussian kernels of model $\mM$ are uniformly continuous in total variation and, in particular, Assumption \ref{A2}.ii) holds.

Since the assumptions of Corollary \ref{cor:GANS} hold for linear and Gaussian models, it follows that there is a sequence $\gamma_{1:T}$ such that nudging using the map in \eqref{eq:NudgingmapGaussian} yields an increased marginal likelihood, $p_T(y_{1:T}|\mM^\alpha) \ge p_T(y_{1:T}| \mM)$. The computer simulations in Section \ref{ssKalman} show that it is not difficult to find sequences of steps $\gamma_{1:T}$ that improve the marginal likelihood.

\begin{remark}\label{remark: UniconTV}
  Note that the linearity of the mean is not required in Eq. \eqref{eq:TVNormNormal}. Specifically, for any $\mu_1$ and $\mu_2$, using Eq. \eqref{eq:SpectralDecQ}, we obtain
    \begin{equation}
\label{eq:continuousTV}
\norm{\mathcal{N}(\sd x_t; \mu_1,Q_t)-\mathcal{N}(\sd x_t;\mu_2,Q_t)}_{TV}\leq 
\frac{1}{2}\norm{\Lambda_t^{-\frac{1}{2}}U_t}\norm{\mu_1-\mu_2},  
\end{equation}
where $Q_t=U_t^\top \Lambda_t U_t$ is the eigenvalue decomposition of $Q_t$.
Therefore, any Gaussian transition kernel (not just the linear ones) is uniformly continuous in total variation, in the sense of Assumption  \ref{A2}. 
\end{remark}

\begin{remark}
Consider the linear-Gaussian observation model $Y_t = a I_{d_x} X_t + V_t$ where $V_t \sim \mathcal{N}(0, \sigma^2 I_{d_y})$ and $a \neq 0$. With the choice of the step-size $\gamma_t=\gamma^\star = (\sigma/a)^2$, we obtain that $x_t^\star = \alpha_t(x, \gamma^\star) = (1/a) y_t$ which is the maximiser of the likelihood $g_t$, i.e., $x_t^\star \in \arg \max_x g_t(x)$ for every $t$. This creates the degenerate kernel that is mentioned in Remark~\ref{rem:maximiser_mle_nudging}. Such cases should be avoided in practice. Note, however, that for more general observation models, the problem $\arg \max_x g_t(x)$ is intractable and thus this issue is less prominent.
\end{remark}

\section{Computer simulations} \label{sNumerical}


\subsection{Nudging in a linear-Gaussian state-space model} \label{ssKalman}

\subsubsection{Simulation setup}

Let us consider a linear-Gaussian SSM, which is tractable as shown in Section~\ref{NKalman}. In particular, we consider a four-dimensional controlled linear dynamical system similar to the setup in \cite{akyildiz2020nudging}. Let $I_n$ denote the identity matrix of dimension $n$, we define
\begin{align}
\pi_0(\mathsf{d} x_0) &= \mathcal{N}(\mathsf{d} x_0; \mu_0, P_0), \\
K^\star(x_{t-1},\mathsf{d}x_t) &= \mathcal{N}(\mathsf{d}x_t; Ax_{t-1} + B L (x_{t-1} - x_{\star}), Q), \\
g_t(y_t|x_t) &\propto \exp\{-\frac{1}{2} (y_t-C x_t)^\top R^{-1} (y_t-C x_t)\},
\end{align}
 where we choose $C=I_4$, 
\begin{align*}
A = \begin{bmatrix}
I_2 & \kappa I_2 \\
0 & I_2
\end{bmatrix},
\quad
B = \begin{bmatrix}
    0 & I_2
\end{bmatrix}^\top, \quad Q = \begin{bmatrix}
\frac{\kappa^3}{3} I_2 & \frac{\kappa^2}{2} I_2 \\
\frac{\kappa^2}{2} I_2 & \kappa I_2
\end{bmatrix},
\end{align*}
with $\kappa = 0.04$ and
\begin{align}
L = \begin{bmatrix}
-0.0134 & 0.0 & -0.0381 & 0.0 \\
0.0 & -0.0134 & 0.0 & -0.0381
\end{bmatrix}.
\end{align}
This system defines a \textit{controlled} linear dynamical system that moves the system towards the target state $x_{\star} = [140, 140, 0, 0]^\top$ where $L$ is found by solving a Riccati equation \cite{bertsekas2012dynamic}. Since this policy would not be known a priori to an observer interested in filtering the observations from this system, we  explore the use of nudging together with the \textit{misspecified} SSM with the transition kernel
\begin{align}
K(x_{t-1}, \mathsf{d} x_t) = \mathcal{N}(\mathsf{d}x_t; A x_{t-1}, Q),
\end{align}
which ignores the control terms in $K^\star$. We next define the nudged kernel (see Eq. \eqref{eq:nudged_Gaussian_kernel})
\begin{align}
K^\alpha(x_{t-1}, \mathsf{d}x_t) = \mathcal{N}(\mathsf{d}x_t; M A x_{t-1} + \gamma C^\top R^{-1} y_t, M Q M^\top),
\end{align}
where we choose a fixed step size $\gamma > 0$ and
\begin{align*}
M = I_4 - \gamma C^\top R^{-1} C.
\end{align*}

%
\subsubsection{Numerical results}

Numerical results for the linear-Gaussian SSM can be seen from Fig.~\ref{fig: likelihood_stepsizes} and Fig.~\ref{fig: NMSE_Kalman}. In particular, Fig.~\ref{fig: likelihood_stepsizes} demonstrates the behaviour of the log marginal likelihoods w.r.t. varying step-sizes within the step-size range $\gamma \in [5 \times 10^{-3}, 1.5 \times 10^{-1}]$. We note that since all considered models within this section are linear Gaussian SSMs, the log marginal likelihood computations are exact. 

It can be seen from Fig.~\ref{fig: likelihood_stepsizes} that the log marginal likelihoods of the nudged KF can be slightly higher than the log marginal likelihood of the original KF with the correct parameters. This numerically verifies the result we obtained in Corollary~\ref{cor:NudgingMLE}, empirically demonstrating the \textit{nudging gain} (one should note, however, that the result in Corollary~\ref{cor:NudgingMLE} is a result w.r.t. the MLE, rather than the \textit{true} parameter). 

Next, Fig.~\ref{fig: NMSE_Kalman} shows  a similar performance w.r.t. the normalised mean square errors (NMSEs) rather than the log-marginal likelihoods. The NMSE at discrete time $t$ is constructed as 
\beq \label{eq: NMSE} 
\text{NMSE}_t = \frac{\norm{x_t-\hat x_t}_2^2}{\frac{1}{T}\sum_{t=1}^T \norm{x_t}_2^2}, 
\eeq
where $x_t$ is the actual $3$-dimensional state of the system and $\hat x_t$ is its estimate computed by the PF. The NMSE for each simulation is then computed as the mean over time of these errors, namely, $\text{NMSE}= \frac{1}{T}\sum_{t=1}^T \text{NMSE}_t$. Similar to Fig. \ref{fig: likelihood_stepsizes}, we observe that nudging yields much lower NMSEs than the misspecified KF. However, expectedly, in terms of NMSEs w.r.t. the ground truth states, the KF with the correct parameters remains the best estimator.

\begin{figure}[htb] 
  \begin{minipage}[t]{0.48\linewidth}
    \centering
    \includegraphics[width=1.1\linewidth]{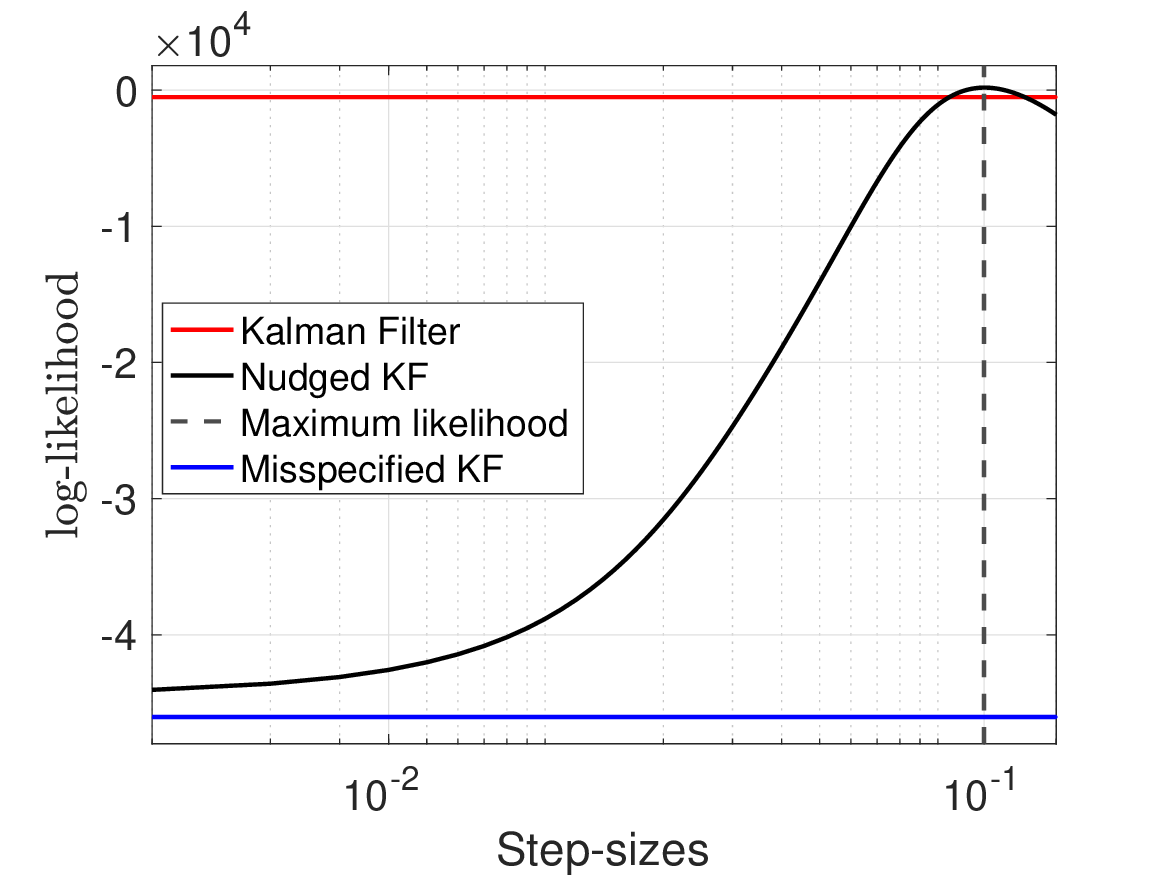} 
    \caption{Comparison of marginal likelihoods for the step-size interval $\gamma \in [5 \times 10^{-3}, 1.5 \times 10^{-1}]$ where $\gamma_t := \gamma$ for all $t = 1, \ldots, T$.  The figure shows that the nudged Kalman filter attains a higher likelihood than the original (correctly specified) Kalman filter for a range of step-size values and attains much higher likelihood than the misspecified Kalman filter across all step-sizes.} \label{fig: likelihood_stepsizes}
    \vspace{4ex}
  \end{minipage}\hfill
  \begin{minipage}[t]{0.48\linewidth}
    \centering
    \includegraphics[width=1.1\linewidth]{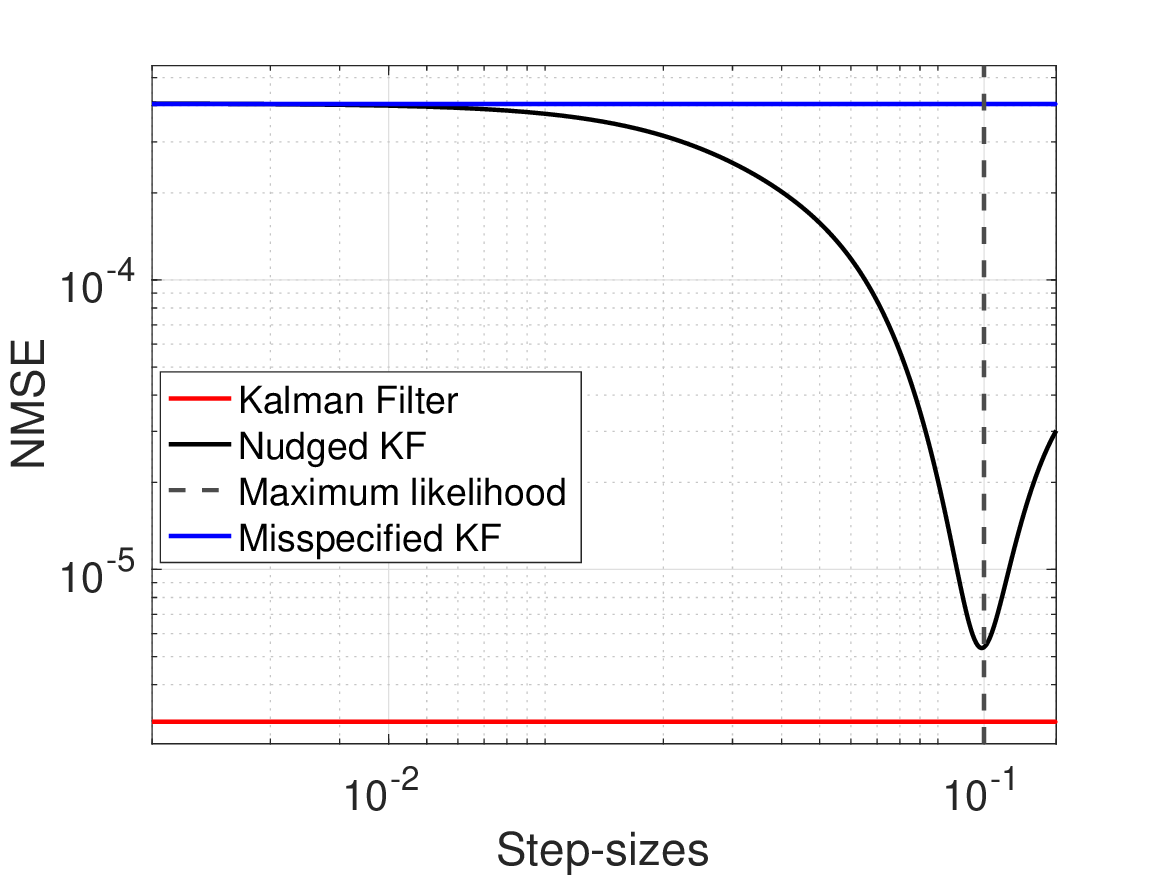} 
    \caption{Comparison of the NMSEs for the step-size interval $\gamma \in [5 \times 10^{-3}, 1.5 \times 10^{-1}]$ where $\gamma_t := \gamma$ for all $t = 1, \ldots, T$. The figure shows, similarly, the nudged Kalman filter attains a lower NMSE than the misspecified Kalman filter.} \label{fig: NMSE_Kalman}
    \vspace{4ex}
  \end{minipage} 

\end{figure}

\subsection{Stochastic Lorenz 63 model}

\subsubsection{Simulation setup}

We examine the problem of tracking the dynamic variables of a $3$-dimensional Lorenz system with additive dynamical noise and partial noisy observations. The system dynamics are governed by a stochastic differential equation (SDE).
Specifically, consider a stochastic process $\{\tilde X(s)\}_{s \in (0,\infty)}$ taking values on $\mbR^3$, described  by the system of Itô SDEs
\begin{align}
   &\sd\tilde X_1 = -S(\tilde X_1-Y_1) + \sd W_1, \label{SDEL1}\\  &\sd \tilde X_2 = -R \tilde X_1-\tilde X_2- \tilde X_1 \tilde X_3 + \sd W_2, \label{SDEL2}\\
   &\sd \tilde X_3 = \tilde X_1 \tilde X_2-B \tilde X_3 + \sd W_3,\label{SDEL3}
\end{align}
where $\{ W_i(s)\}_{s\in (0,\infty)}, \; i=1,2,3$, are independent one-dimensional Wiener processes, $s$ denotes continuous time, and $\{S,R,B\}\in \mbR$ are constant model parameters. A discrete-time approximation of this system can be derived using the Euler-Maruyama method with a time step $h>0$, resulting in the difference equations
\begin{align}
 &\tilde X_{1,n}=\tilde X_{1,n-1}-h S(\tilde X_{1,n-1}-\tilde X_{2,n-1})+\sqrt{h} U_{1,n}, \label{eq:LorenzX1} \\  
 &\tilde X_{2,n}=\tilde X_{2,n-1}-h (R \tilde X_{1,n-1}-\tilde X_{2,n-1}-\tilde X_{1,n-1}\tilde X_{3,n-1})+\sqrt{h} U_{2,n}, \label{eq:LorenzX2} \\  
 &X_{3,n}=\tilde X_{3,n-1}-h (\tilde X_{1,n-1}\tilde X_{2,n-1}- B \tilde X_{3,n-1})+\sqrt{h} U_{3,n}, \label{eq:LorenzX3}
\end{align} 
where $n=1,2,...,$ is discrete time,
and $\{U_i\}_{n},$  $i=1,2,3$, are independent sequences of i.i.d. $\mathcal{N}(0,1)$ random variables. 

We assume that the system  is observed every $n_0\geq 1$ discrete-time steps. Specifically, we assume that only the variable $\tilde X_{1,n}$ is observed, meaning that we collect a sequence of one-dimensional observations $\{Y_t\}_{t=1,2,...},$ of the form
\beq \label{eq:LorenzY1}
Y_{t} = \tilde X_{1,n_0t} + V_{t},
\eeq
where $\{ V_{t}\}_{t=1,2,...}$ is a sequence of i.i.d. r.v.'s with distribution $\mathcal{N}(0,\sigma^2)$.

Let us denote $X_{i,t}= \tilde X_{i,n_0t}$, so that the $t$-th observation can be written as 
\beq \label{eq:40th-observations}
Y_t = X_{1,t}+V_t
\eeq 
and $X_t = ( X_{1,t},X_{2,t},X_{3,t})^\mathsf{T}$ denotes the state of the system at discrete time $t$ (or continuous time $s=hn_0t$).  
The iteration of Eq. \eqref{eq:LorenzX1}-\eqref{eq:LorenzX3} yields the Markov kernel $K_t(x_{t-1}, \sd x_t)$ while Eq. \eqref{eq:40th-observations} yields the (Gaussian) likelihood function $g_t(x_t)$. We assume a Gaussian prior distributions $\pi_0(\sd x_0)= \mathcal{N}(\hat x_0,\hat C_0)$, where $\hat x_0 = (1,1,1)^\top , C_0 = \sigma_0 I $, and $\sigma_0=20$. The model is parameterised by the constant vector $\theta=(S,R,B)^\top$. In particular, the transition kernel depends on $\theta$ and we write $K_t(x_{t-1},\sd x_t)\equiv K_{t,\theta}(x_{t-1},\sd x_t).$ The resulting parametric model is denoted $\mM_\theta= \{\pi_0,K_\theta,g\}.$  To simulate the state signal and synthetic observations from model $\mM_\theta$, we select the commonly used standard parameter values
\beq\label{eq:LorenzParameter}
\theta^*=(S,R,B)^\mathsf{T}=\left(10,28,\frac{8}{3}\right)^\mathsf{T},
\eeq
which make the deterministic Lorenz 63 chaotic. We assume that the step size for the Euler method is $h=10^{-3}$ and the system is observed every $n_0=40$ discrete time steps. For each simulation, we run the system for $t=1,...,T$, where $T=500$. This amounts to a simulation of the original SDE \eqref{SDEL1}-\eqref{SDEL3} over the continuous time interval $[0,Tn_0h]=[0,20]$.

We apply the gradient ascent nudging method of Section \ref{cor:GANS}, where the transformation $\alpha_t(x,\gamma)$ is defined in \eqref{eq:NudgingGradientAscent}. 
The nudging kernel $K_t^\alpha(x_{t-1}, \sd x_t)$ can be sampled in two steps: 
\vspace{.2cm}
\begin{enumerate}[i)]\label{}
    \item Draw $\hat x_t$ from the original kernel $K_t(x_{t-1}, \sd x_t)$ (this is done by iterating Eqs. \eqref{eq:LorenzX1}-\eqref{eq:LorenzX3} with initial condition $x_{t-1}$).
    \item Apply the correction $x_t=\alpha_t(\hat x_t, \gamma_t).$
\end{enumerate}
\vspace{.2cm}
The nudged model is denoted by $\mM_\theta^\alpha=\{\pi_0,K_\theta^\alpha, g \}.$
It is easy to see that in this case, the gradient  $\nabla \log g_t(x)$ is Lipschitz
with constant $L_t=1/\sigma^2, \; t=1,...,T,$ therefore, as mentioned in Section \ref{SGAN}, we can select $\gamma_t = \gamma \in (0, 2\sigma^2)$ (the parameter $\gamma_t$ is constant for all $1 \le t \le T$). 

Note that we should not set $\gamma = \sigma^2$ in this particular case, since that choice leads to degenerate nudged kernels $K_t^\alpha$ as described in Remark \ref{rem:maximiser_mle_nudging}.

We approximate numerically the $\log$ of the Bayesian evidence for models $\mM_\theta$ and $\mM_\theta^\alpha$, i.e., the quantities $p_T(y_{1:T}\vert \mM_\theta)$ and $p_T(y_{1:T}\vert \mM_\theta^\alpha)$, respectively, by running standard PFs \cite{gordon1993novel} (see also \cite{doucet2000sequential}, \cite{djuric2003particle} and \cite{bain2008fundamentals}) with a sufficiently large number of particles $N,$ for each model $\mM_\theta$ and $\mM_\theta^\alpha$ with the same sequence of observations $Y_{1:T}=y_{1:T}$. If the PF yields a sequence of equally weighted particles sets $\{x^i_t \}_{i=1}^N$ for $t=1,..,T$, then the Monte Carlo estimate of the $\log$ Bayesian evidence is 
$$\log p_T(y_{1:T}\mid \cdot) \approx \log p_T^N(y_{1:T}\mid \cdot)=\sum_{t=1}^T \log \frac{1}{N} \sum_{i=1}^N g_t(x_t^i).$$

\subsubsection{Numerical results}

In order to test whether the proposed nudging scheme can ensure an increased log marginal likelihood in a practical setup (i.e., with fixed step size $\gamma$) 
we have run $200$ independent simulation of a PF for the models $\mM_{\theta}$ and $\mM_{ \theta}^\alpha$ (that is, we generate the state, the observations, and the PF estimates across 200 independent trials), using the parameter $\theta$ in Eq. \eqref{eq:LorenzParameter}. The number of particles is $N = 500$, and the initial condition $x_0$ is randomly drawn from the distribution $\mathcal{N}(\hat{x}_0, C_0)$ with $\hat{x}_0=(1,1,1)^\top$ and $C_0= 20 I$. The observation variance, defined in Eq. \eqref{eq:40th-observations}, is $\sigma^2 = 1$, and we choose the step size $\gamma =  0.8\sigma^2,$ constant for each time step $t$.

\begin{figure}[htb] 
  \begin{minipage}[t]{0.48\linewidth}
    \centering
    \includegraphics[width=1.1\linewidth]{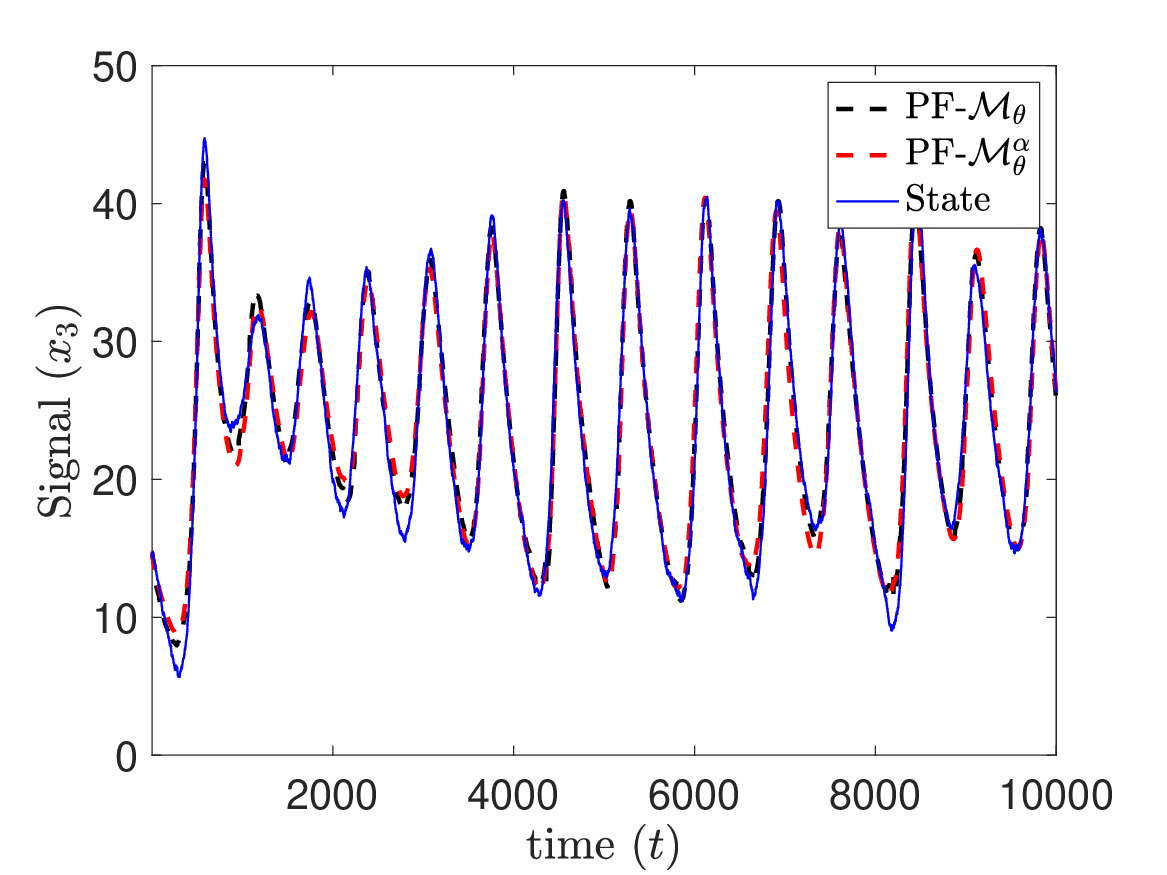} 
    \caption{Coordinate $x_1$ of the state and its PF estimates with models $\mM_\theta$ and $\mM_\theta^\alpha$.} \label{fig:LorenzX1}
    \vspace{4ex}
  \end{minipage}\hfill
  \begin{minipage}[t]{0.48\linewidth}
    \centering
    \includegraphics[width=1.1\linewidth]{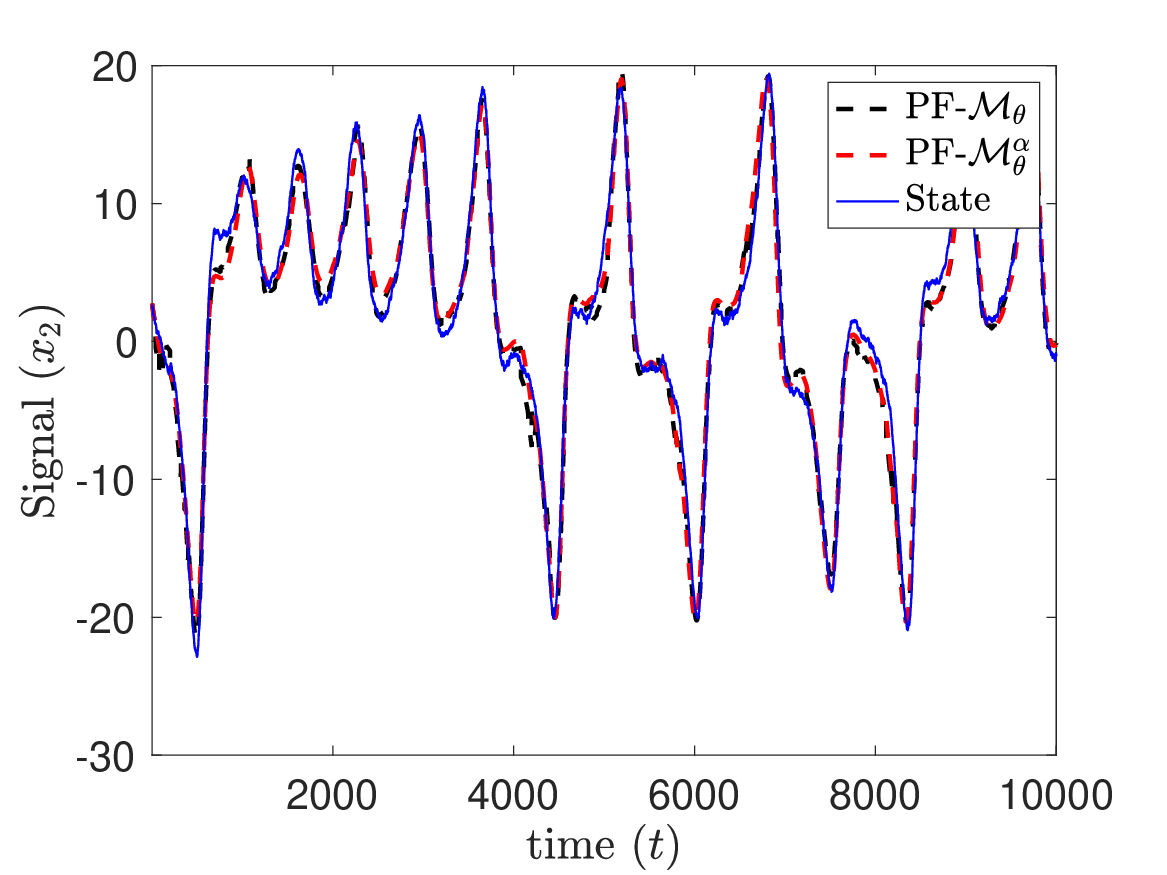} 
    \caption{Coordinate $x_2$ of the state and its PF estimates with models $\mM_\theta$ and $\mM_\theta^\alpha$.} \label{fig:LorenzX2}
    \vspace{4ex}
  \end{minipage} 
  \begin{minipage}[t]{0.48\linewidth}
    \centering
    \includegraphics[width=1.1\linewidth]{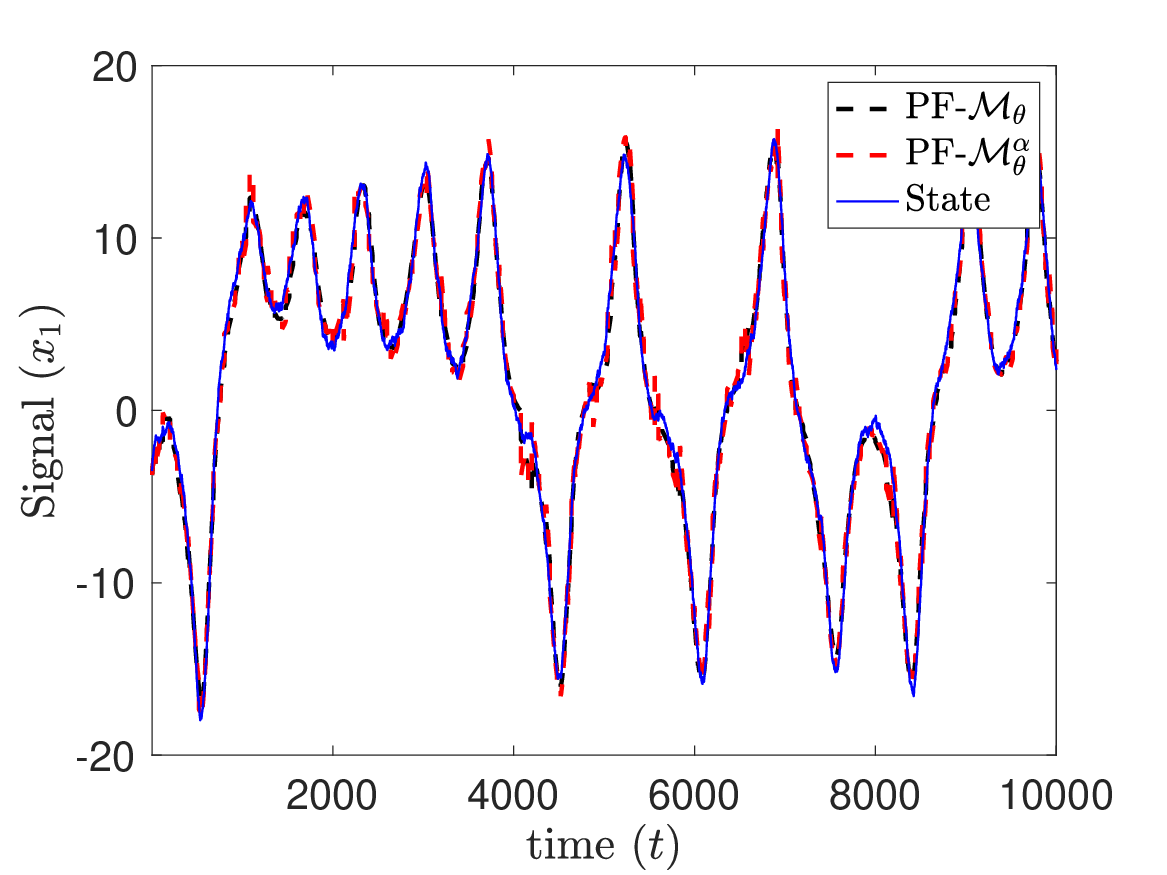} 
    \caption{Coordinate $x_3$ of the state and its PF estimates with models $\mM_\theta$ and $\mM_\theta^\alpha$.} \label{fig:LorenzX3}
    \vspace{4ex}
  \end{minipage} \hfill
  \begin{minipage}[t]{0.48\linewidth}
    \centering
    \includegraphics[width=1.1\linewidth]{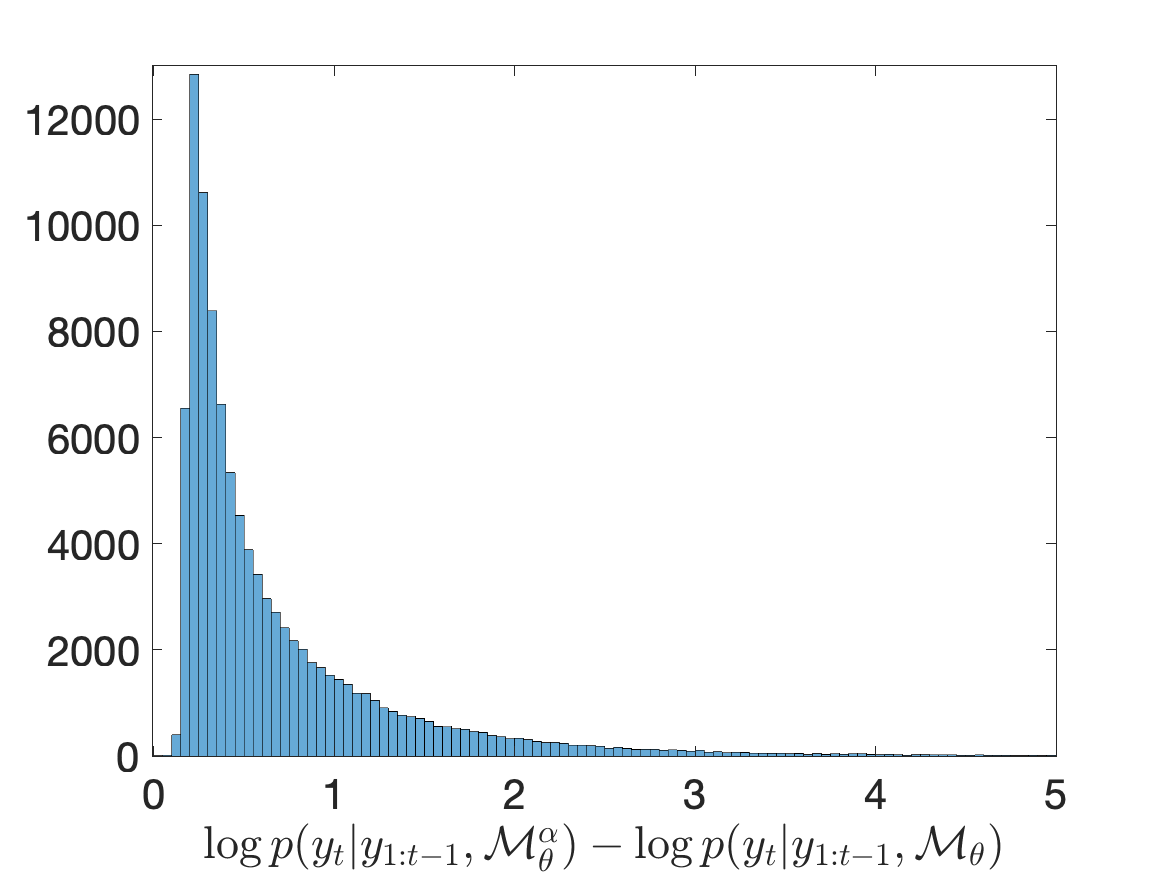}
    \caption{Difference in $\log$ incremental likelihood of $\mM_\theta^\alpha$ and $\mM_\theta$, $t=1,...,T.$} \label{fig:HistogramIncrementalLikelihood} 
    \vspace{4ex}
  \end{minipage} 
 
\end{figure}

Figures \ref{fig:LorenzX1}, \ref{fig:LorenzX2} and \ref{fig:LorenzX3} illustrate the time evolution of the state of the stochastic Lorenz 63 model and their estimates computed via a standard PF for both models, $\mM_\theta$ and $\mM^\alpha_\theta$. Although the approximations look similar, a closer examination reveals significant differences in performance. Specifically, Figure \ref{fig:HistogramIncrementalLikelihood} provides a histogram of the differences between the incremental likelihoods of the two models, expressed as $$\log p(y_t \mid y_{1:t-1}, \mM^\alpha_\theta) - \log p(y_t \mid y_{1:t-1}, \mM_\theta),$$ across the time steps $t=1, \dots, T$. 
This histogram shows a consistent positive difference at each time step, suggesting that the nudged model, $\mM^\alpha_\theta$, reliably enhances the Bayesian evidence. This implies, in particular, that the overall log likelihood, $\log p_T(y_{1:T
} \mid \mM^\alpha_\theta)$, is greater than $\log p_T(y_{1:T
} \mid \mM_\theta)$. Therefore, the nudged model $\mM_\theta^\alpha$ not only approximates the state similarly to the original model $\mM_\theta$ but also provides an improvement in terms of compatibility with the observed data.

Figure \ref{fig:BoxplotM} shows box plots of  the log likelihoods $\log p_T(y_{1:T} \vert \mM_\theta)$ and $\log p_T(y_{1:T} \vert \mM_\theta^\alpha)$ obtained in the same experiment. We observe that the empirical distribution has a larger median for the model with nudging $ \mM_\theta^\alpha$ and the $25\%$ and $75\%$ percentiles are also higher compared to the results  with the original model $\mM_\theta.$ For the same set of simulations, Figure \ref{fig:AverageBayesM} shows the average values of $\log p_t(y_{1:t} \vert \mM_\theta)$ and $\log p_t(y_{1:t} \vert \mM_\theta^\alpha)$ versus the observation index $t=1,...,T.$ Again, we see that nudging improves the $\log$-likelihood. Specically, $\log p_t(y_{1:t} \vert \mM_\theta^\alpha)\geq \log p_t(y_{1:t} \vert \mM_\theta)$ for every $t$.
\begin{figure}[htb] 
  \label{ fig7} 
  \begin{minipage}[t]{0.48\linewidth}
    \centering
    \includegraphics[width=1.1\linewidth]{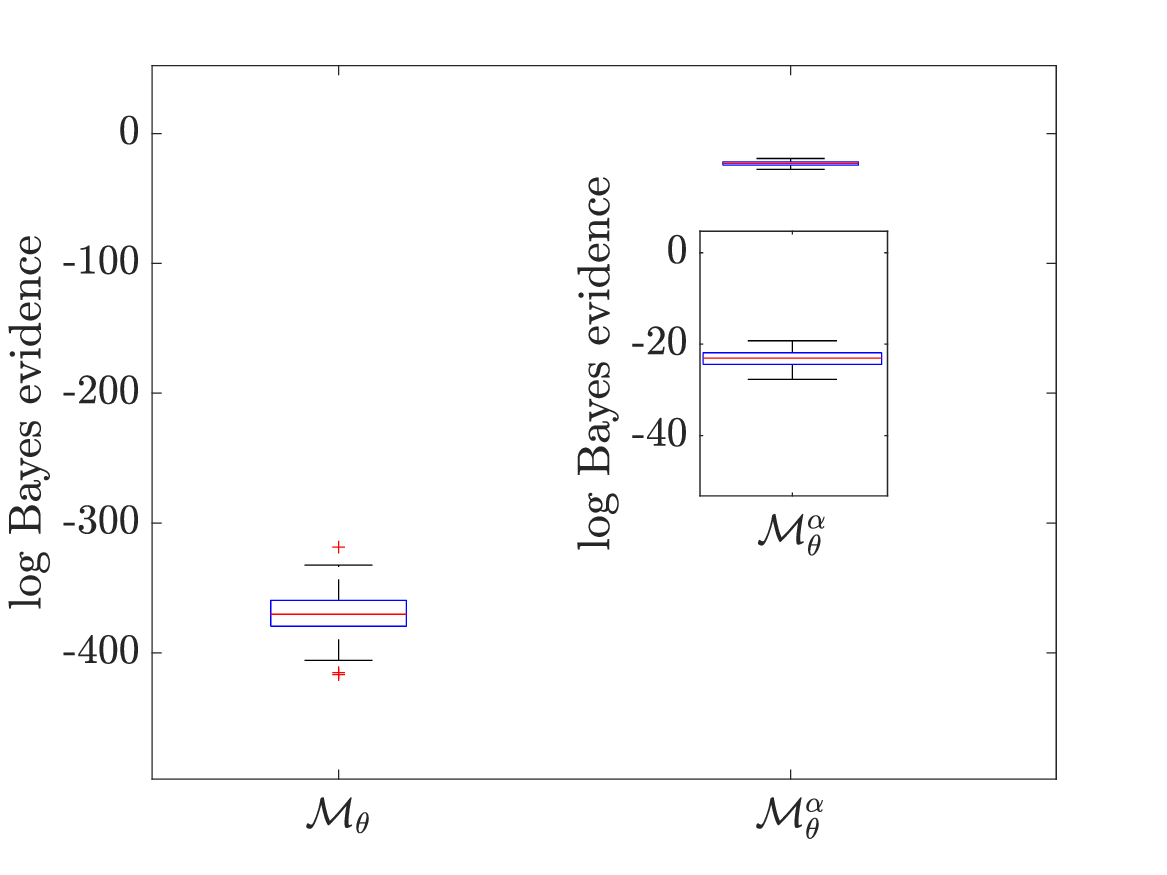} 
    \caption{Box plot of the estimated Bayesian evidence for $\mM_\theta$ and $\mM_\theta^\alpha$, over 200 independent simulations.} \label{fig:BoxplotM}
    \vspace{4ex}
  \end{minipage}\hfill
  \begin{minipage}[t]{0.48\linewidth}
    \centering
    \includegraphics[width=1.1\linewidth]{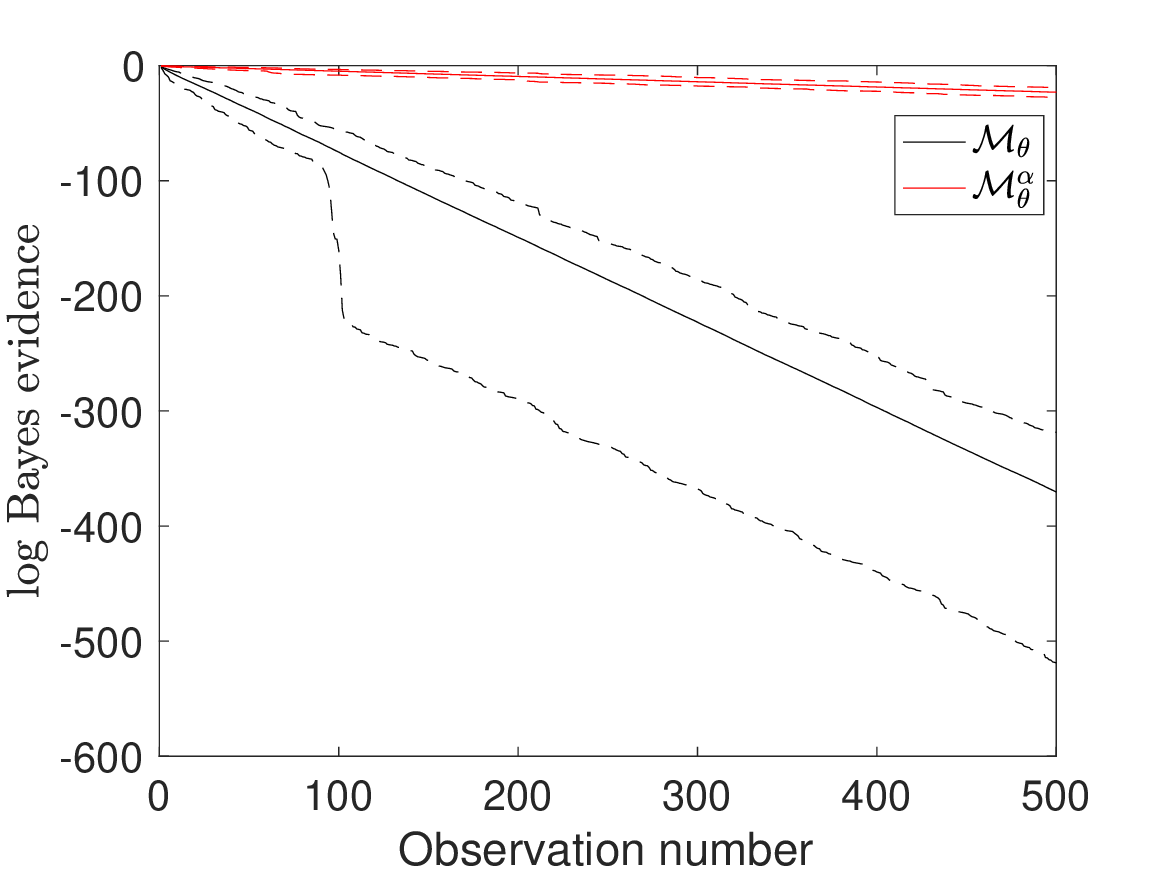} 
    \caption{Average over 200 independent simulations of the $\log$ Bayesian evidence vs  
observations, with maximum and minimum values (dashed lines)for the models $\mM_\theta$ and $\mM_\theta^\alpha$.} \label{fig:AverageBayesM}
    \vspace{4ex}
  \end{minipage} 

\end{figure}

In the next computer experiment, we examine how 
a parameter mismatch affects the performance of the filter. 
For this purpose we consider three models: 
\begin{itemize}
    \item The original model $\mM_\theta$, with $\theta$ as in Eq. \eqref{eq:LorenzParameter}. this model is used to generate the signal and the observations $Y_{1:T}=y_{1:T}$ in each independent simulation.
    \item A mismatched model $\mM_{\tilde \theta}$, where $\tilde \theta=(10,28,\frac{8}{3}+\epsilon)^\top$ and $\epsilon=\frac{11}{5}.$ For each simulation we run a PF on this model, using the data $Y_{1:T}=y_{1:T}$ generated with the original model $\mM_\theta.$
    \item The nudged model $\mM_{\tilde \theta}^\alpha$. For each simulation, we also run a PF on $\mM_{\tilde \theta}^\alpha$, with the same data $Y_{1:T}=y_{1:T}$ generated from $\mM_\theta$.
\end{itemize}

We have run $200$ independent simulations with the setup described above. For all simulations the number of particles is $N = 500$, and the initial condition $x_0$ is randomly drawn from the distribution $\mathcal{N}(\hat{x}_0, C_0)$ with $\hat{x}_0=(1,1,1)^\top$ and $C_0= 20 I$. The observation variance is $\sigma^2 = 1$, and we choose the fixed step size $\gamma =  0.8\sigma^2.$

 Due to the chaotic dynamics of the system, the parameter mismatch significantly impacts the dynamics, causing the PF built upon $\mM_{\tilde{\theta}}$ to lose track of the state signals. 
 However, tracking remains effective in the PF built upon the nudged model $\mM^\alpha_{\tilde{\theta}}$, as illustrated in Figures \ref{fig:LorenzX1MM} to \ref{fig:LorenzX3MM}.

\begin{figure}[hbt]  
  \begin{minipage}[t]{0.48\linewidth}
    \centering
    \includegraphics[width=1.1\linewidth]{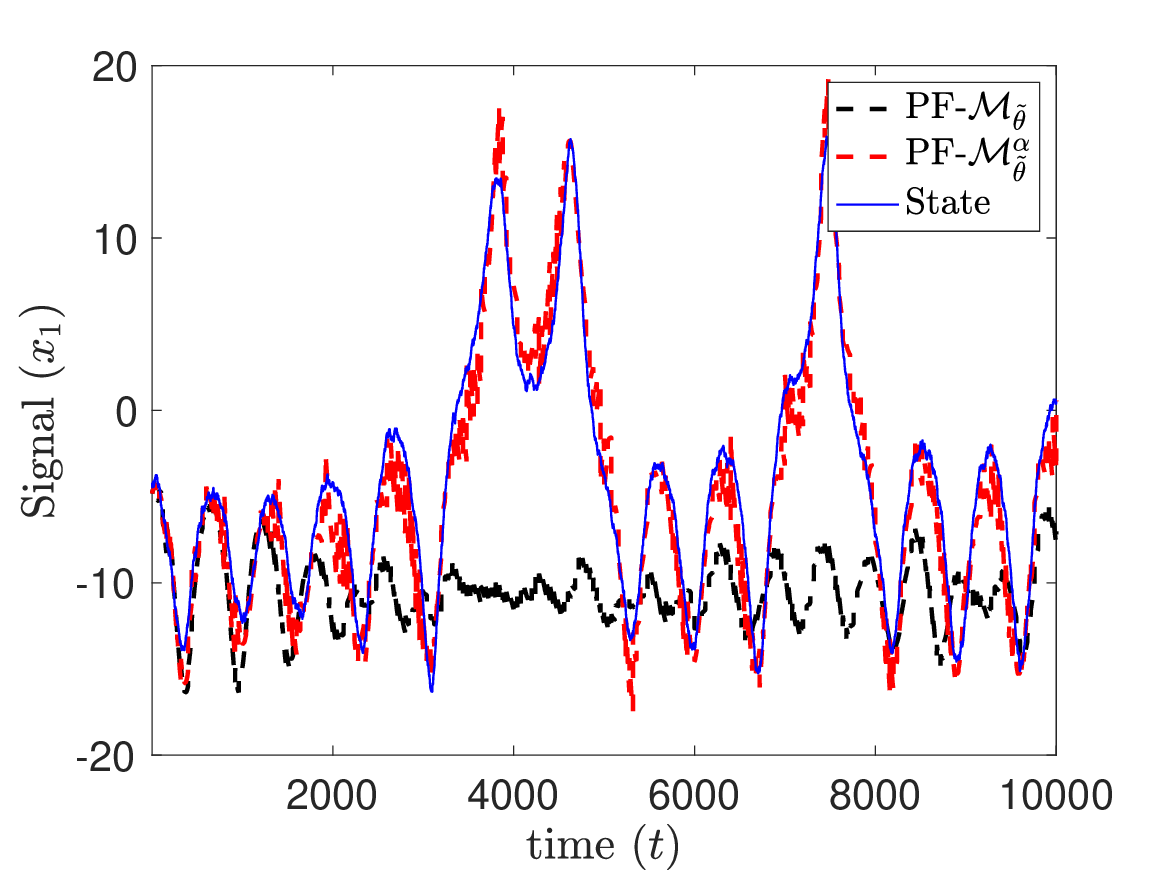} 
    \caption{Coordinate $x_1$ of the state and its PF estimates for $\hspace{.1cm} \mM_{\tilde\theta} \hspace{.1cm}$ and $\hspace{.1cm}\mM_{\tilde \theta}^\alpha$,  $\hspace{.3cm} \tilde \theta=(S,R,B+\epsilon)^\mathsf{T}$.} \label{fig:LorenzX1MM}
    \vspace{4ex}
  \end{minipage}\hfill
  \begin{minipage}[t]{0.48\linewidth}
    \centering
    \includegraphics[width=1.1\linewidth]{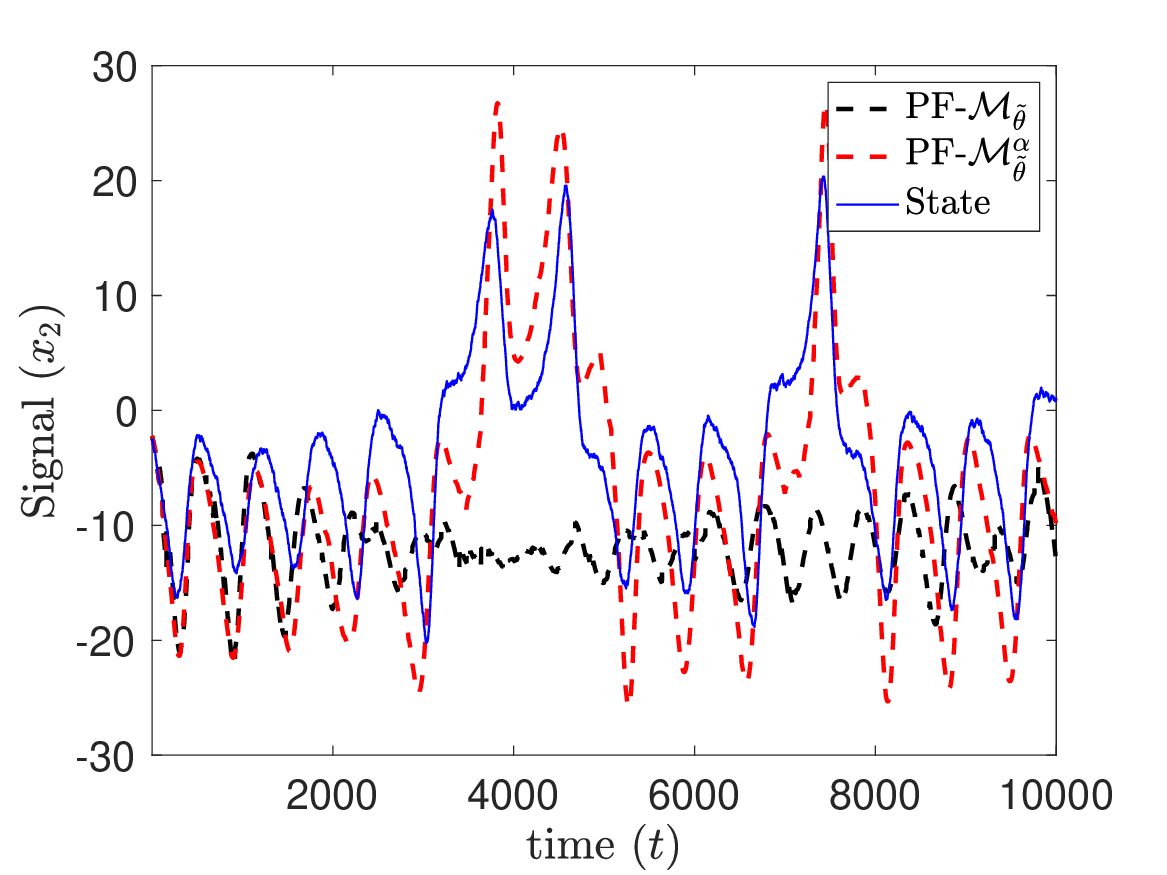} 
    \caption{Coordinate $x_2$ of the state and its PF estimates for $\mM_{\tilde\theta}$ and $\mM_{\tilde \theta}^\alpha$, $\tilde \theta=(S,R,B+\epsilon)^\mathsf{T}$.} \label{fig:LorenzX2MM}
    \vspace{4ex}
  \end{minipage} 
  \begin{minipage}[t]{0.48\linewidth}
    \centering
    \includegraphics[width=1.1\linewidth]{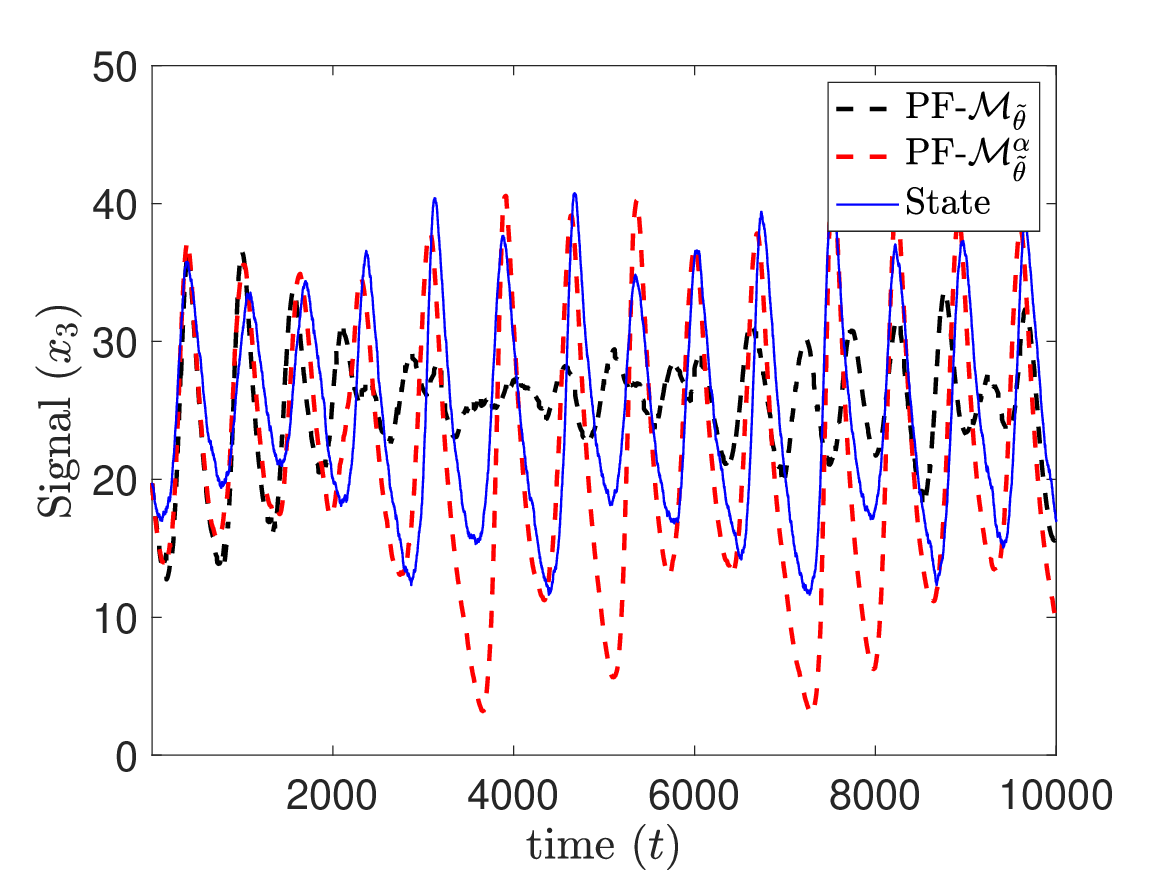} 
    \caption{Coordinate $x_3$ of the state and its PF estimates for $\mM_{\tilde\theta}$ and $\mM_{\tilde \theta}^\alpha$, $\tilde \theta=(S,R,B+\epsilon)^\mathsf{T}$.} \label{fig:LorenzX3MM}
    \vspace{4ex}
  \end{minipage} \hfill
  \begin{minipage}[t]{0.48\linewidth}
    \centering
    \includegraphics[width=1.1\linewidth]{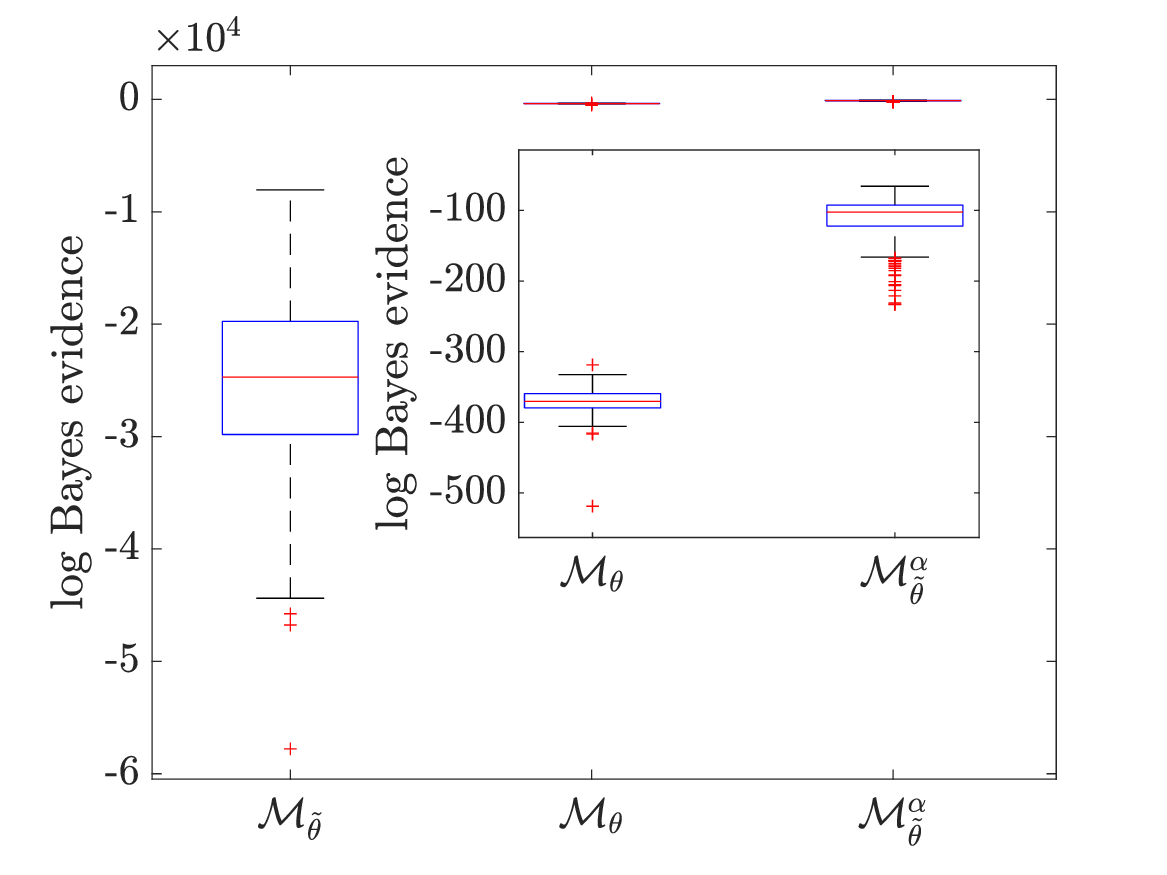} 
    \caption{Box plot comparison for the estimated $\log$ Bayesian evidence for $\mM_{\tilde \theta}$, $\mM_{\tilde \theta}^\alpha$, and the true model $\mM_\theta$. The inset graph is a zoom view of the box plots for $\mM_\theta$ and $\mM_{\tilde \theta}^\alpha.$} \label{fig:BoxplotMM}
    \vspace{4ex}
  \end{minipage} 
\end{figure}

For the same set of simulations, 
Figure \ref{fig:BoxplotMM} presents box plots of the empirical distribution of the $\log$ Bayesian evidence $p_T(y_{1:T}\vert \mM_{\tilde \theta})$ for the mismatched model, and we additionally compare it with the evidence for the true model, $p_T(y_{1:T}\vert \mM_{\theta})$, and the mismatched nudged model, $p_T(y_{1:T}\vert \mM_{\tilde \theta}^\alpha).$
We see that the log Bayesian evidence of the mismatched nudged model $\mM_{\tilde \theta}^\alpha$ is much higher than the evidence of the mismatched model $\mM_{\tilde \theta}$ and even slightly higher than the evidence of the ``true'' model $\mM_\theta$. This shows that nudging can effectively compensate for parameter mismatches.

In our final experiment, we introduce significant mismatches across all values of the parameter vector $\theta$ to evaluate the filter performance under extreme conditions. For this purpose we consider the parameter vector $\hat{\theta}=2\theta$, with $\theta$ as in Eq. \eqref{eq:LorenzParameter}. We assume $2$-dimensional observations for this simulations, namely
\begin{align*}
Y_{1,t} = X_{1,t} + V_{1,t},\\
Y_{2,t} = X_{2,t} + V_{2,t}, 
\end{align*}  and denote $Y_t=(Y_{1,t},Y_{2,t})^\top$ 
where $\{ V_{i,t}\}_{t=1,2,...},$ $i=1,2,$ are sequences of i.i.d. $\mathcal{N}(0,\sigma^2)$ r.v.'s. 
As in previous experiments, we use the original model $\mM_\theta$, with $\theta$ as in Eq. \eqref{eq:LorenzParameter} to generate the signal and the observations $Y_{1:T}=y_{1:T}$ in each simulation.

We have run $200$ independent simulations of the standard trial PF for the models $\mM_{\hat \theta}$ and $\mM_{\hat \theta}^\alpha$, using the parameter value $\hat \theta=2\theta$,  with $\theta$ as in Eq.\eqref{eq:LorenzParameter}. For all simulations the number of particles is $N = 500$, and the initial condition $x_0$ is randomly drawn from the distribution $\mathcal{N}(\hat{x}_0, C_0)$ with $\hat{x}_0=(1,1,1)^\top$ and $C_0= 20 I$. The observation variance is $\sigma^2 = 1$, and we choose the fixed step size $\gamma =  0.8\sigma^2.$ 

\begin{figure}[hbt]  
  \begin{minipage}[t]{0.48\linewidth}
    \centering
    \includegraphics[width=1.1\linewidth]{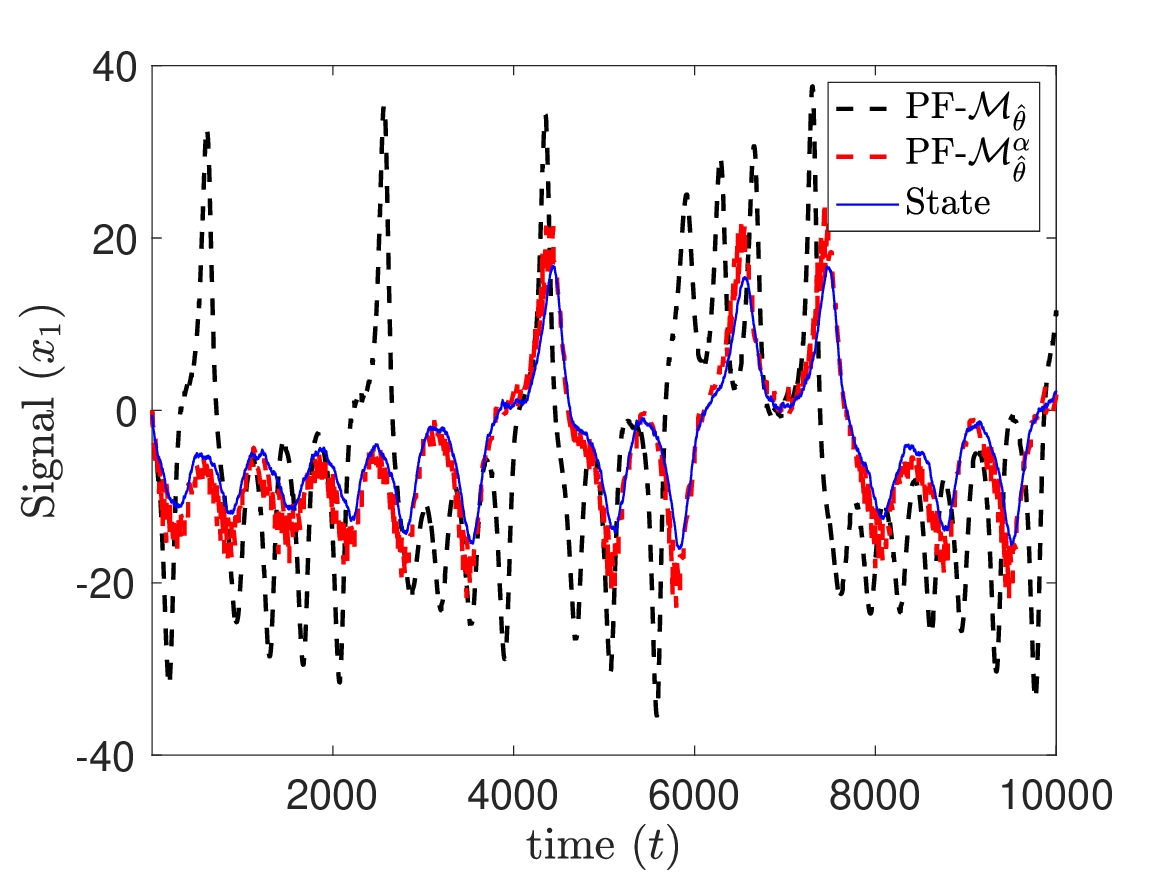} 
    \caption{Coordinate $x_1$ of the state and its PF estimates for models $\mM_{\hat\theta}$ and $\mM_{\hat \theta}^\alpha$, with $\hat \theta = (2S, 2R, 2B)^\mathsf{T}$.} \label{fig:LorenzX1MM100}
    \vspace{4ex}
  \end{minipage}\hfill
  \begin{minipage}[t]{0.48\linewidth}
    \centering
    \includegraphics[width=1.1\linewidth]{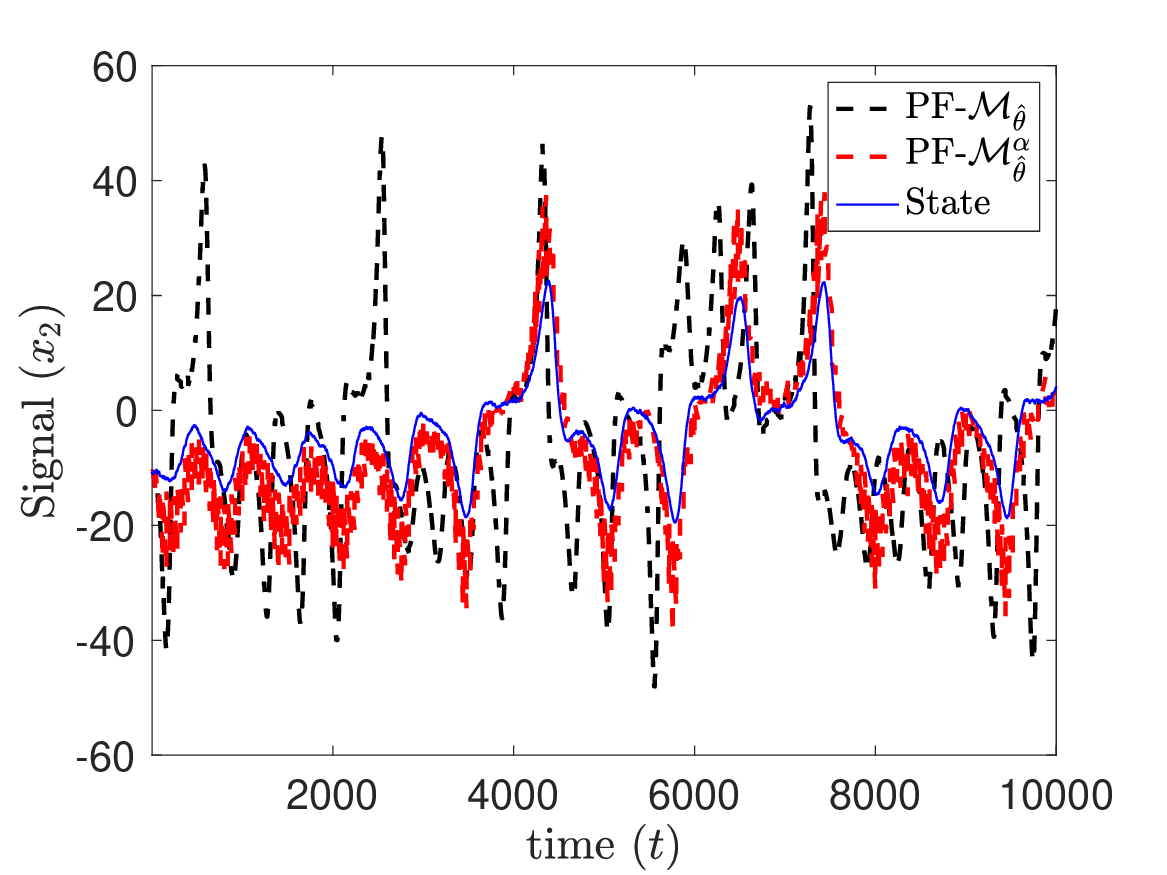} 
    \caption{Coordinate $x_2$ of the state and its PF estimates for models $\mM_{\hat\theta}$ and $\mM_{\hat \theta}^\alpha$, with $\hat \theta = (2S, 2R, 2B)^\mathsf{T}$.} \label{fig:LorenzX2MM100}
    \vspace{4ex}
  \end{minipage} 
  \begin{minipage}[t]{0.48\linewidth}
    \centering
    \includegraphics[width=1.1\linewidth]{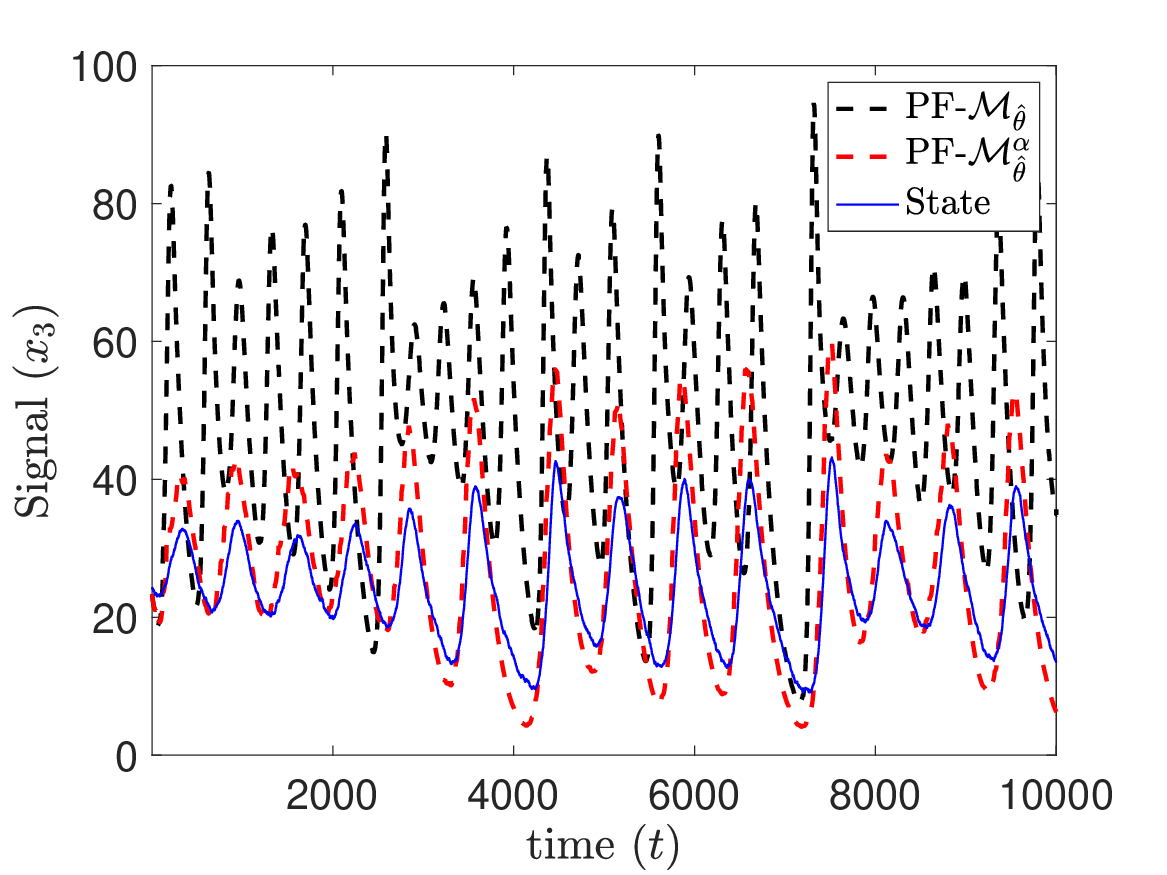} 
    \caption{Coordinate $x_3$ of the state and its PF estimates for models $\mM_{\hat\theta}$ and $\mM_{\hat \theta}^\alpha$, with $\hat \theta = (2S, 2R, 2B)^\mathsf{T}$.} \label{fig:LorenzX3MM100}
    \vspace{4ex}
  \end{minipage} \hfill
  \begin{minipage}[t]{0.48\linewidth}
    \centering
    \includegraphics[width=1.1\linewidth]{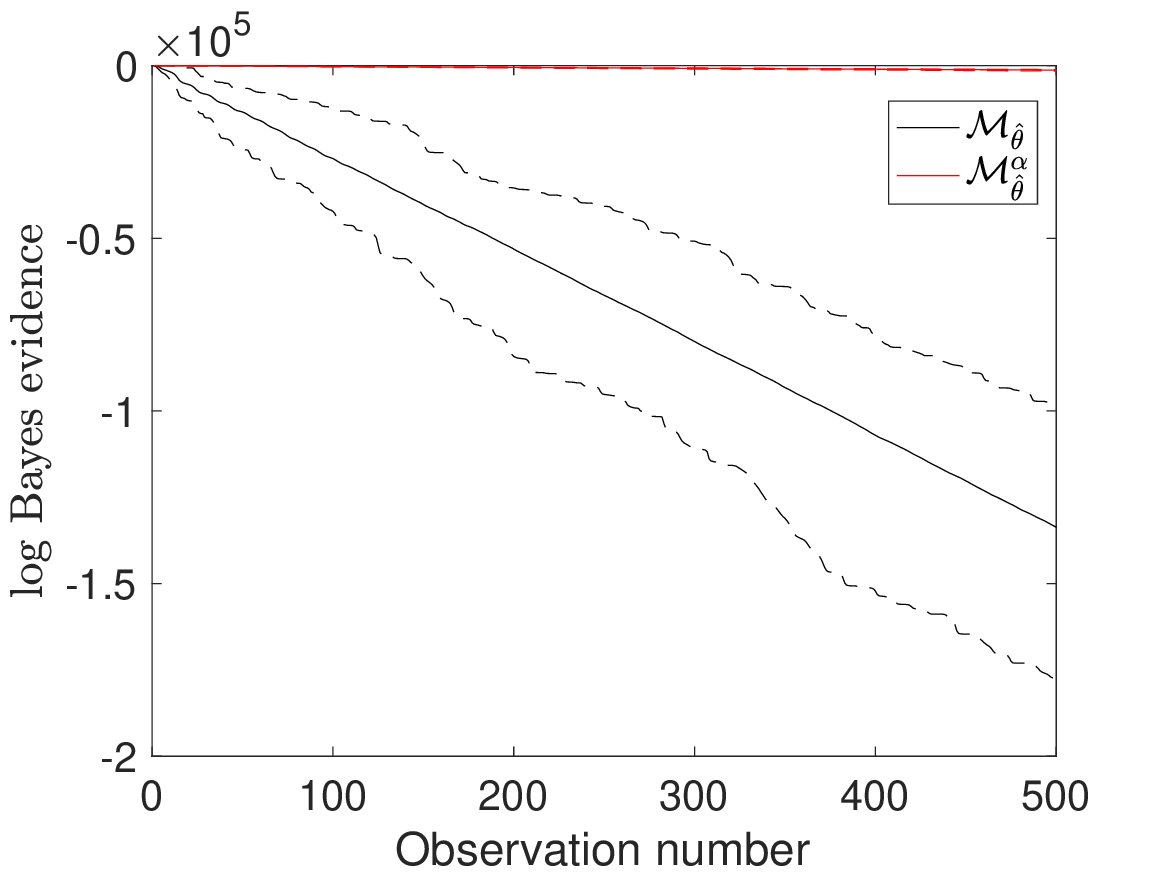} 
    \caption{Average over 200 independent simulations of the $\log$ Bayesian evidence vs. observation number with maximum and minimum values (dashed lines) for $\mM_{\hat{\theta}}$ and $\mM^\alpha_{\hat{\theta}}$.} \label{fig:AverageBayesMM100}
    \vspace{4ex}
  \end{minipage}

\end{figure}

 Figures \ref{fig:LorenzX1MM100} to \ref{fig:LorenzX3MM100} illustrate how the PF for the model $\mM_{\hat{\theta}}$ fails to track the state under extreme parameter mismatches. In contrast, the PF for the nudged model $\mM^\alpha_{\hat{\theta}}$ continues to track the state reliably. 
Figure \ref{fig:AverageBayesMM100} presents the average $\log$ Bayesian evidence as a function of $t$. Note that the Bayesian evidence for $\mM^\alpha_{\hat{\theta}}$ remains consistently higher and more stable over time compared to the model without nudging, indicating a stronger alignment with the observed data in this extreme setup.

Finally, Table \ref{Table} summarises the NMSE and the Bayesian model evidence (at the final time step $T$) obtained from our three previous experiments. Specifically, the table displays, for each model, the average NMSE and the average $\log$ Bayesian evidence (or $\log$ marginal likelihood) computed over $200$ independent simulations. The sample standard deviation is shown between brackets. We observe that nudging always increases the Bayesian model evidence and, in the case of parameters mismatches, it also reduces the NMSE very significantly. 

\begin{table}[htb]
\begin{tabularx}{\textwidth} { 
  | >{\raggedright\arraybackslash}X 
  | >{\centering\arraybackslash}X 
  | >{\raggedleft\arraybackslash}X | }
 \hline \centering
\textbf{Model}& \centering\textbf{ NMSE} & $\log$\textbf{ - Bayesian evidence} 

\\
 \hline
 $\mM_{\theta} \quad \theta=(S,R,B)^\mathsf{T}$  & $ 0.0040 \; (0.00073)$  & $
 -370.4164 \; (  19.1346)$ \\
\hline
 $\mM_{\theta}^\alpha$  & $  0.0078 \; ( 0.00190)$  & $-23.1279 \; (  1.7278)$  \\
\hline
 $\mM_{\tilde\theta} \quad \tilde \theta=(S,R,B_\epsilon)^\mathsf{T}$ & $ 0.4314 \; (0.1144)$  & $ -2.5016\times 10^{4} \; ( 8.1299 \times 10^{3})$  \\
\hline
$\mM_{\tilde\theta}^\alpha$  & $ 0.1487 \;  ( 0.0471)$   & $-114.7217 \;  ( 34.1360)$  \\
\hline
$\mM_{\hat \theta} \quad \hat \theta=(2S,2R,2B)^\mathsf{T}$  & $1.7484 \;  (  0.1226)$   & $-1.3366\times 10^{5}\; (   1.4343\times 10^{4})$  \\
\hline
$\mM_{\hat \theta}^\alpha$  & $ 0.1190 \; ( 0.0043)$   & $ -1.2961\times 10^{3}\; (   77.6686)$  \\
\hline
\end{tabularx}
\caption{ NMSE and the $\log$ Bayesian evidence at the final time step for both the true parameter and mismatched cases. We display the sample mean over $200$ simulations, with the standard deviation between brackets. Note that for model $\mM_{\hat \theta}^\alpha$ the observations are 2-dimensional, versus 1-dimensional for $\mM_{\tilde \theta}^\alpha$, which explains the reduction in NMSE despite the larger error in the parameters.}
\label{Table}
\end{table}

\section{Conclusions} \label{sConclusions}

We have introduced a general methodology for nudging in SSMs that consists in a data-driven modification of the Markov kernels in the model. We have proved that the resulting nudged models can attain (when adequately implemented) a better agreement with the available data --as quantified by the marginal likelihood or Bayesian model evidence. Although other possibilities exist, we have paid especial attention to an implementation of the methodology using the gradient of the log-likelihood of the state of the SSM, since this quantity is often available and, when analytically intractable, it can be approximated numerically. 

The application of the proposed methodology has been illustrated both analytically and numerically. In particular, we have looked into the specific cases of linear-Gaussian SSMs and (possibly nonlinear) SSMs indexed by a parameter vector. We have particularised the theoretical guarantees of the methodology to these two scenarios and we have presented numerical results obtained through computer simulations of a 4-dimensional linear-Gaussian model and a stochastic Lorenz 63 model with partial observations. Both sets of computer simulations show that the proposed nudging schemes are easy to implement. Also, they appear particularly effective in compensating for erroneous dynamical drifts due to mismatches in the SSM parameters.

A potential pitfall of the methodology is the degeneracy of the nudged Markov kernels that occurs when the nudging transformation maximises the likelihood of the state. This issue has been identified and it is straightforward to avoid in practice (e.g., by choosing smaller nudging steps). Further research is needed in order to quantify the gain in the Bayesian evidence obtained by specific nudging schemes, to analyse alternative (non-gradient-based) implementations and to assess the efficiency of the methodology in relevant real-world problems.

%
\backmatter

\section*{Declarations}

\begin{itemize}
\item \textbf{Funding} JM and FG acknowledge the support of the Office of Naval Research (award N00014-22-1-2647) and Spain's \textit{Agencia Estatal de Investigaci\'on} (ref. PID2021-125159NB-I00 TYCHE) funded by MCIN/AEI/10.13039/501100011033 and by ``ERDF A way of making Europe".

The work of DC has been partially supported by European Research Council (ERC) Synergy grant STUODDLV-856408.

\item \textbf{Conflict of interest} The authors have no conflict of interest to declare.
\item \textbf{Ethics approval}  Not applicable.
\item \textbf{Consent for publication} Not applicable.
\item \textbf{Data availability} All codes used for the numerical results are available upon request from authors.  
\end{itemize}

\noindent

\bmhead{Acknowledgements}

This work was partially carried out during a visit of FG to the Department of Mathematics at Imperial College London, from October 2023 to March 2024.

%
%
%
%


\begin{appendices}

%
\section{An alternative nudging model}\label{App1}

We define a different way to perform the nudging that is easier to analyze
and preserves the same Bayesian evidence. The original model is $\mM=(\pi_0,K,g )$, where $\pi_0(\sd x_0)$ is the initial probability distribution, $K  = \{ K_t\}_{ t\geq 1}$ is the family of
Markov kernels for the process $X_t$ and $g = \{ g_t\}_{ t\geq 1}$ is the family of likelihoods generated
by the observations $\{ Y_t = y_t\}_{t\geq 1}$. Let $\mX$ be the state space and let $\alpha_t:\mX\mapsto\mX$ be the nudging function. We have adopted the nudged model $\mM^\alpha=(\pi_0,K^\alpha,g)$, where the nudged kernel is defined in \eqref{eq:NudgingKernel} as
$$
K_t^\alpha(x_{t-1},\sd x_t) := \int \delta_{\alpha_t(x_t')}(\sd x_t) K_t(x_{t-1},\sd x_t'), \quad \xi_t^\alpha = K_t^\alpha \pi_{t-1}^\alpha,
$$
and for an integrable test function $f:\mX\mapsto\mbR$, $\pi_t^\alpha(f)$ is defined in \eqref{eq:Nudgingpostpre}.

We introduce the alternative model $\bar \mM^\alpha=(\pi_0,\bar K^\alpha,g^\alpha)$, where 
\beq
\bar K_t^\alpha(x_{t-1},\sd x_t) := K_t( \alpha_{t-1}(x_{t-1}), \sd x_t), \quad t=1,2,...,
\label{eq:DefKalpha2}
\eeq
$\alpha_0(x):=x$ is the identity function,
and 
$g_t^\alpha := g_t \circ \alpha_t$ ($\circ$ denotes composition of functions). Then $\bar \xi_t^\alpha = \bar K_t^\alpha \bar \pi_{t-1}^\alpha$, where for any test function $f:\mX\mapsto\mbR$
\beq
\bar\pi_t^\alpha(f) = \frac{\bar\xi_t^\alpha(fg_t^\alpha)}{\bar\xi_t^\alpha(g_t^\alpha)}.
\label{eq:BarPif}
\eeq

\begin{lemma} \label{lm2models}
For any $t\geq 1$, and  any integrable test function $f:\mX\mapsto \mbR$, 
\beqa
\xi_t^\alpha(f) &=& \bar \xi_t^\alpha(f \circ \alpha_t), \quad \text{and} \label{eqLem1}\\
\pi_t^\alpha(f) &=& \bar \pi_t^\alpha(f \circ \alpha_t). \label{eq:Lem2}
\eeqa
\end{lemma}

\begin{proof}
We proceed by induction. At time $t=1$ we have $\xi_1^\alpha(f) = K_1^\alpha\pi_0$ and using the definition of $K_t^\alpha$ in \eqref{eq:NudgingKernel}, we obtain
\beq
\xi_1^\alpha(f) = K_1^\alpha\pi_0(f) = \int\int f(x_1) \int \delta_{\alpha_1(x_1')}(\sd x_1) K_1(x_0,\sd x_1')\pi_0(\sd x_0)
\nn
\eeq
which, integrating w.r.t. the delta measure, yields
\beqa
\xi_1^\alpha(f) &=& \int \int (f \circ \alpha_1)(x_1') K_1(x_0,\sd x_1') \pi_0(\sd x_0) \nn\\
&=& K_1\pi_0(f \circ \alpha_1) = \xi_1(f \circ \alpha_1)
\label{eq:Xi1}
\eeqa
Moreover, since by definition $\alpha_0(x)=x$ (the identity function), we readily find that $\bar\xi_1^\alpha=K_1\pi_0=\xi_1$, hence \eqref{eq:Xi1} implies $\xi_1^\alpha(f) = \bar\xi_1^\alpha(f \circ \alpha_1)$ and the identity \eqref{eqLem1} holds at time $t=1$.

Combining \eqref{eq:BarPif} and $\bar\xi_1^\alpha=\xi_1$ we obtain
\beq
\bar\pi_1^\alpha(f \circ \alpha_1) = \frac{ \xi_1\left( (fg_1)\circ \alpha_1 \right) }{ \xi_1(g_1\circ\alpha_1) } 
= \frac{ \xi_1^\alpha\left( f g_1 \right) }{ \xi_1^\alpha(g_1) } 
= \pi_1^\alpha(f),
\nn
\eeq
where the second equality follows from \eqref{eq:Xi1} and the third one follows from \eqref{eq:Nudgingpostpre}. Hence, also equation \eqref{eq:Lem2} holds at time $t=1$.

For the induction step, let us assume that 
\beq
\pi_{t-1}^\alpha(f)=\bar\pi_{t-1}^\alpha(f\circ \alpha_{t-1}).
\label{eq:IndHip}
\eeq
At time $t$, we obtain
\beq
\xi_t^\alpha(f) = K_t^\alpha\pi_{t-1}^\alpha(f) = \int\int f(x_t) \int \delta_{\alpha_t(x_t')}(\sd x_t) K_t(x_{t-1},\sd x_t')\pi_{t-1}^\alpha(\sd x_{t-1})
\nn
\eeq
and, integrating w.r.t. the delta measure, we have
\beqa
\xi_t^\alpha(f) &=& \int \int (f \circ \alpha_t)(x_t') K_t(x_{t-1},\sd x_t') \pi_{t-1}^\alpha(\sd x_{t-1}) \nn\\
&=& K_t\pi_{t-1}^\alpha(f \circ \alpha_t).
\label{eq:Xit}
\eeqa
For the alternative model, on the other hand, we arrive at
\beqa
\bar\xi_t^\alpha(f) &=& \bar K_t^\alpha \bar\pi_{t-1}^\alpha(f) \nn\\
&=& \int \int f(x_t) K_t\left( \alpha_{t-1}(x_{t-1}), \sd x_t \right) \bar\pi_{t-1}^\alpha(\sd x_{t-1})\nn\\
&=& \int \left( \bar f_t \circ \alpha_{t-1} \right)(x_{t-1}) \bar\pi_{t-1}^\alpha(\sd x_{t-1})\nn\\
&=& \bar\pi_{t-1}^\alpha\left( \bar f_t \circ \alpha_{t-1} \right),\label{eq:BarPift}
\eeqa
where the second equality follows from the definition of $\bar K_t^\alpha$ in \eqref{eq:DefKalpha2} and we have introduced the notation $\bar f_t(x) := \int f(x_t) K_t(x,\sd x_t)$ in the third equality. The induction hypothesis \eqref{eq:IndHip} together with \eqref{eq:BarPift} yields
\beq
\bar\xi_t^\alpha(f) = \pi_{t-1}^\alpha(\bar f_t) = K_t\pi_{t-1}^\alpha(f)
\label{eq:BarXit}
\eeq
and, comparing \eqref{eq:BarXit} above and \eqref{eq:Xit}, we readily find that $\xi_t^\alpha(f) = \bar\xi_t^\alpha(f \circ \alpha_t)$ and, hence, equation \eqref{eqLem1} in the statement of Lemma \ref{lm2models} holds for arbitrary time $t$.

As for the laws $\pi_t^\alpha$ and $\bar\pi_t^\alpha$, equations \eqref{eq:Nudgingpostpre} and \eqref{eq:Xit} taken together yield
\beq
\pi_t^\alpha(f) = \frac{
    K_t\pi_{t-1}^\alpha\left( (fg_t) \circ \alpha_t \right)
}{
    K_t\pi_{t-1}^\alpha\left( g_t \circ \alpha_t \right)
},
\label{eq:Pitf}
\eeq
while combining \eqref{eq:BarPif} and \eqref{eq:BarXit} we arrive at
\beq
\bar\pi_t^\alpha(f) = \frac{
    K_t\pi_{t-1}^\alpha\left( f(g_t \circ \alpha_t) \right)
}{
    K_t\pi_{t-1}^\alpha\left( g_t \circ \alpha_t \right)
}.
\label{eq:BarPitf}
\eeq
Comparing \eqref{eq:Pitf} and \eqref{eq:BarPitf} we readily see that $\bar\pi_t^\alpha(f\circ \alpha_t) = \pi_t^\alpha(f)$. Therefore, equation \eqref{eq:Lem2} in the statement of Lemma \ref{lm2models} holds for all $t$.
\end{proof}

\begin{remark}\label{re: NudgingInverse}
If the map $\alpha_t$ is invertible, then
\beq
\bar\xi_t^\alpha(f) = \xi_t^\alpha(f \circ \alpha_t^{-1}) \quad \text{and} \quad
\bar\pi_t^\alpha(f) = \pi_t^\alpha(f \circ \alpha_t^{-1}). \nn
\eeq
So, in general, we can recover $\pi_t^\alpha$ and $\xi_t^\alpha$ from $\bar\pi_t^{\alpha}$ and $\bar\xi_t^{\alpha}$, but not necessarily the other way around. 
\end{remark}

\begin{remark}\label{SBER}
From equations \eqref{eq:Pitf} and \eqref{eq:BarPitf} we observe that the normalisation constants for $\pi_t^\alpha$ and $\bar\pi_t^\alpha$ are the same. As a consequence, both models have the same Bayesian evidence, i.e. 
 \beq \label{eq:Bayeseviden2Nudgingmodels} p_T(y_{1:T}\vert\bar\mM^\alpha)=\prod_{i=1}^T\bar\xi_i^\alpha(g_i^\alpha) =  \prod_{i=1}^T\xi_i^\alpha(g_i)=p_T(y_{1:T}\vert \mM^\alpha).
 \eeq
\end{remark}

%
\section{Proof of Theorem \ref{ISSL}}\label{App2}

From Remark \ref{SBER} the nudging model $\bar \mM^{\alpha}:=\{\pi_0,\bar K^{\alpha}, g^{\alpha}\}$ defined in \eqref{eq:DefKalpha2}
and the nudging model $\mM^{\alpha}=\{\pi_0, K^{\alpha}, g \}$ given by the change in the transition kernel in (\ref{eq:NudgingKernel}) have the same Bayesian evidence (see Eq. \eqref{eq:Bayeseviden2Nudgingmodels} above). Hereafter, we aim at proving that 
    $$p_T(y_{1:T}\vert \bar \mM^{\alpha})\geq p_T(y_{1:T}\vert  \mM).$$ 
We proceed with a series of preliminary results in Section \ref{B1}, while the key induction argument of the proof is presented in Section \ref{B2}

\subsection{Preliminary results}\label{B1}
Let $\{\alpha_t\}_{t\in \mbN}$ be a family of parametric nudging transformations as defined in Definition \ref{def:NudParT}.
 We adopt the simplified notation 
  \beq 
  \bar \xi_t^{\alpha} \equiv \bar \xi_t^{\alpha(\gamma_{1:t-1})}, \quad \bar \pi_t^{\alpha}\equiv \bar \pi_t^{\alpha(\gamma_{1:t})},
  \eeq
  where $\alpha(\gamma_{1:t})$ represents the composition of the transformations with the likelihood functions $g_i$ and the kernels $K_i$, as defined in  \eqref{eq:DefKalpha2} from $i=1,...,t$.
  This simplification is intended to make the analysis easier to read. However, it is important to keep in mind that the predictive measure $\bar \xi_t^{\alpha}$ depends on the sequence $\gamma_{1:t-1}$, while filter $\bar \pi_t^{\alpha}$ depends on the sequence $\gamma_{1:t}$.

At each time step $t$, we can quantify the differences  between the normalisation constants of the models $\bar \mM^\alpha$ and $\mM$ in terms of the predictive measures $\xi_t$ and $\bar \xi_t^{\alpha}$ as well as the total increment of the function $g_t$, given by $\Delta_{g_t}(\gamma)=\xi_t(g_t^{\alpha}-g_t)$ as follows.

\begin{proposition}\label{prop: B1}
    For $t\in \mbN$ and $\gamma\in [0,\Gamma_t]$, we have $$\xi_t(g_t)+\Delta_{g_t}(\gamma)\leq \norm{g_t}_{\infty}\norm{\xi_t-\bar \xi_t^{\alpha}}_{TV}+\bar \xi_t^{\alpha}(g_t^{\alpha}).$$
\end{proposition}
\begin{proof}
    Note that we can write
    \beq\label{eq:DiffUpdates}
 \xi_t(g_t)+\Delta_{g_t}(\gamma)=\int_{\mX} g_t(\alpha_t(x,\gamma))\xi_t(\sd x).
 \eeq
 Now, for any  given sequence $\gamma_{0:t-1} $, adding and subtracting $ \int_{\mX} g_t(\alpha_t(x,\gamma))\bar \xi_t^{\alpha}(\sd x)$, on the right hand side of \eqref{eq:DiffUpdates}
 we obtain
    
\begin{align*}
\xi_t(g_t)+\Delta_{g_t}(\gamma) &=  \int_{\mX} g_t(\alpha_t(x,\gamma))(\xi_t-\bar \xi_t^{\alpha})(\sd x)+\int_{\mX} g_t(\alpha_t(x,\gamma))\bar \xi_t^{\alpha}(\sd x),\\
& \leq \norm{g_t}_{\infty}\norm{\xi_t-\bar \xi_t^{\alpha}}_{TV} +\bar \xi_t^{\alpha}(g_t^{\alpha}).
    \end{align*}
\end{proof}

The previous proposition implies immediately the following corollary.
\begin{corollary}\label{CNI}
   If the parameter $\gamma$ is selected in such a way that $$\Delta_{g_t}(\gamma)\geq \norm{g_t}_{\infty}\norm{\xi_t-\bar \xi_t^{\alpha}}_{TV}, \text{ then } \; \xi_t(g_t)\leq \bar \xi_t^{\alpha}(g_t^{\alpha}).$$
\end{corollary} 
Therefore, choosing the sequence of parameters $\gamma_t$ to ensure that $\Delta_{g_t}(\gamma_t)\geq \norm{g_t}_{\infty}\norm{\xi_t-\bar \xi_t^{\alpha}}_{TV},\text{ for } t=1,...,T$, is sufficient
to ensure that
 $p_T(y_{1:T}\vert \bar \mM^{\alpha})\geq p_T(y_{1:T}\vert \mM).$
Hence, it is natural to seek a method to guaranteed that control the error introduced in the predictive measures by the nudging transformation. This can be achieved in several steps. First, we consider the error introduced by the  modified likelihood functions $g_t^\alpha$.

\begin{definition}\label{UNM}
   Let $\mu_t$ and $\mu_t^\alpha$ be the non-normalised finite measures constructed as
    $$\mu_t(F):=\int_{F} g_t(x)\xi_t(\sd x), \quad \mu_t^{\alpha}(F):=\int_{F} g_t^{\alpha}(x)\bar \xi_t^{\alpha}(\sd x), \quad \forall F \in \mF.$$ 
\end{definition}
\begin{proposition}\label{rec}
 Let Assumption \ref{A2}. i) hold. Then, the non-normalised measures $\mu_t \text{ and }\mu_t^{\alpha}$ satisfy the inequality
  $$\norm{\mu_t-\mu_t^{\alpha}}_{TV}\leq \Delta_{g_t}(\gamma_t)+\norm{g_t}_{\infty}\norm{\xi_t-\bar \xi_t^{\alpha}}_{TV}.$$
\end{proposition}
\begin{proof}
   For any set $F\in \mF$ 
    $$\mu_t(F)-\mu_t^{\alpha}(F)=\int_F g_t(x)\xi_t(\sd x)-\int_F g_t^{\alpha}(x)\bar \xi_t^{\alpha}(\sd x).$$
   Adding and subtracting $\int_F g_t^{\alpha}(x)\xi_t(\sd x)$ on the right hand side of the equation above yields

    $$\mu_t(F)-\mu_t^{\alpha}(F)=\int_F (g_t(x)-g_t^{\alpha}(x))\xi_t(\sd x)+\int_F g_t^{\alpha}(x)(\xi_t-\bar \xi_t^{\alpha})(\sd x),$$
    hence 
    \beq\label{I}
    \abs{\mu(F)-\mu_t^{\alpha}(F)}\leq \abs{\int_F (g_t(x)-g_t^{\alpha}(x))\xi_t(\sd x)}+\int_F g_t^{\alpha}(x)\abs{\xi_t-\bar \xi_t^{\alpha}}(\sd x).
    \eeq
    Note that, by Eq.  \eqref{eq:IncreasingNudging} we have $g_t(x)\leq g_t^{\alpha}(x),\;\forall x\in \mX$, hence
    $$0\leq \int_F (g_t^{\alpha}(x)-g_t(x))\xi_t(\sd x)\leq \int_\mX (g_t^{\alpha}(x)-g_t(x))\xi_t(\sd x)=\Delta_{g_t}(\gamma), \; \forall F\in \mF,$$
    and, therefore
    \beq\label{I1}
    \abs{\int_F (g_t(x)-g_t^{\alpha}(x))\xi_t(\sd x)}\leq\Delta_{g_t}(\gamma).
    \eeq
  On the other hand, 
   by Assumption \ref{A2}. i) we have 
   $\norm{g_t}_{\infty}<\infty,$ which yields
  \beq\label{I2}
  \int_F g_t^{\alpha}(x)\abs{\xi_t-\bar \xi_t^{\alpha}}(\sd x)\leq\norm{g_t}_{\infty} \norm{\xi_t-\xi_t^\alpha}_{TV}.
   \eeq
   Combining the inequalities \eqref{I1}, \eqref{I2} and \eqref{I} concludes the proof. 
\end{proof}

Next, we need to normalise the measures $\mu_t,$ and $\mu_t^\alpha$ in order to obtain the probability measures $\pi_t:=\mu_t/\xi_t(g_t),$ and $\bar \pi_t^\alpha:=\mu_t^\alpha/\bar \xi_t^\alpha(g_t^\alpha).$ A way to control the discrepancy between $\pi_t$ and $\bar \pi_t^\alpha$ is given by the proposition below.

\begin{proposition}\label{pimu}
    If $\xi_t(g_t)\leq \bar \xi_t^{\alpha}(g_t^{\alpha}),$  then 
    $$\norm{\pi_t-\bar \pi_t^{\alpha}}_{TV}\leq \frac{\norm{\mu_t-\mu_t^{\alpha}}_{TV}}{\xi_t(g_t)}.$$
\end{proposition}

\begin{proof}
For any $F\in \mF$
    \begin{equation*}
    \bar \pi_t^{\alpha}(F)-\pi_t(F)=\frac{\mu_t^{\alpha}(F)}{\bar \xi_t^{\alpha}(g_t^{\alpha})}-\frac{\mu_t(F)}{\xi_t(g_t)}=\frac{1}{\bar \xi_t^{\alpha}(g_t^{\alpha})}\left[  (\mu_t^{\alpha}-\mu_t)(F)+\pi_t(F)({\xi_t(g_t)-\bar \xi_t^{\alpha}(g_t^{\alpha})})\right],     
    \end{equation*}
    where the second equality is obtained by adding and subtracting $\mu_t^\alpha(F)/\xi_t(g_t)$.
Moreover, since $\xi_t(g_t) \leq \bar \xi_t^{\alpha}(g_t^{\alpha})$ we readily obtain the inequality 
    $$(\bar \pi_t^{\alpha}-\pi_t)(F)\leq \frac{1}{{\xi_t(g_t)}}\left[  (\mu_t^{\alpha}-\mu_t)(F)\right],$$
    that holds for any $F\in\mF.$
In particular for $A,B$ a Hahn decomposition of $\mX$ w.r.t $(\bar \pi_t^{\alpha}-\pi_t)$, (i.e, $A\cup B = \mX$, $B=\stcomp{A}$ and $\bar \pi_t^{\alpha}(A)-\pi_t(A)\geq 0$) yields 
\beq\label{II1}
\norm{\bar \pi_t^{\alpha}-\pi_t}_{TV}=(\bar \pi_t^{\alpha}-\pi_t)(A)\leq \frac{1}{\xi_t(g_t)}\left[  (\mu_t^{\alpha}-\mu_t)(A)\right],
\eeq
where
\beq\label{II2}
    (\mu_t^{\alpha}-\mu_t)(A) \leq \norm{\mu_t^{\alpha}-\mu_t}_{TV}.
\eeq
   Substituting \eqref{II2} back into \eqref{II1} concludes the proof.
\end{proof}
Next, we account for the difference between the Markov kernels $K_t$ and $K_t^\alpha$, which we quantify as
\beq\label{DKt}
    \Delta_{K_{t+1}}(\gamma):=\int_{\mX}\norm{K_{t+1}(x,\cdot)-K_{t+1}(\alpha_t(x,\gamma),\cdot)}_{TV}\pi_t(\sd x).
 \eeq   
Additionally, let $\mathcal{D}_{\mX}:=\{\eta (\pi_1-\pi_2) : \eta\in \mbR, \pi_i\in \mP(\mX), \; i=1,2\}$ be, the linear space generated by the differences of probability measures in $\mX$. We can think of the Markov kernel as an operator $K_t:\mathcal{D}_{\mX}\to \mathcal{D}_{\mX}$ and introduce the induced norm
\begin{equation}\label{normK}
\norm{K_t}_{\mathcal{D}_{\mX}}:=\sup_{\substack{\lambda \in \mathcal{D}_{\mX} \\ \lambda \neq 0}} \frac{\norm{K_t( \lambda)}_{TV}}{\norm{\lambda}_{TV}}.
\end{equation}
It is not difficult to prove that 
\beq\label{INK}
\norm{K_t}_{\mathcal{D}_{\mX}}=\sup_{x,x'\in\mX} \norm{K_t(x,\cdot)-K_t(x',\cdot)}_{TV},
\eeq
(see \cite{dobrushin1956central2} Section 3, Eq. (1.5'')). 
\begin{proposition}\label{DobN}
The induced norm of the nudged operator $\bar K_t^{\alpha}$ satisfies the inequality
$$\norm{\bar K_t^{\alpha}}_{\mathcal{D}_{\mX}}\leq \norm{K_t}_{\mathcal{D}_{\mX}},\quad \text{for }\gamma\in[0,\Gamma_t], \; \text{and }t\geq 1.$$
\end{proposition}
\begin{proof}
   We know $\alpha_t:\mX \times \mbR \to \mX$, i.e., the image $\Im(\alpha_t):=\{ y \in \mX : y = \alpha_t(x,\gamma), \; x\in \mX, \gamma \in[0,\Gamma_t] \} \subseteq \mX$. Then from \eqref{INK}, we have
\begin{align*}
\norm{\bar K_t^\alpha}_{\mathcal{D}_{\mX}}&=\sup_{x,x'\in\mX}\norm{K_t(\alpha(x,\gamma),\cdot)-K_t(\alpha(x',\gamma),\cdot)}_{TV}\\
&= \sup_{y,y'\in\Im(\alpha_t)} \norm{K_t(y,\cdot)-K_t(y',\cdot)}_{TV} 
 \leq \sup_{x,x'\in\mX} \norm{K_t(x,\cdot)-K_t(x',\cdot)}_{TV} \\ &=\norm{K_t}_{\mathcal{D}_{\mX}}, \quad \text{for all } \gamma\in[0,\Gamma_t],\; t\geq 1. 
\end{align*}
\end{proof}

With this, we are able to control the discrepancy in the predictive measure at the $t+1$ step as follows.

\begin{lemma}\label{pmrl}If $\xi_t(g_t)\leq \bar \xi_t^{\alpha}(g_t^{\alpha}),$ then
$$ \norm{\xi_{t+1}-\bar \xi_{t+1}^{\alpha}}_{TV} \leq b_t\norm{\xi_t-\bar \xi_t^{\alpha}}_{TV}+ a_t\Delta_{g_t}(\gamma)+\Delta_{K_{t+1}}(\gamma),$$
where 
$a_t=\frac{\norm{K_{t+1}}_{\mathcal{D}_{\mX}}}{\xi_t(g_t)}$, and $b_t= a_t \norm{g_t}_{\infty}.$
\end{lemma}
\begin{proof}\label{pmr}
The difference between the predictive measures $\xi_{t+1},\bar \xi_{t+1}^{\alpha}$ at the time $t+1$  are given by
\beq\label{1}
\xi_{t+1}(\cdot)-\bar \xi_{t+1}^{\alpha}(\cdot)=\int_{\mX} K_{t+1}(x,\cdot)\pi_{t}(\sd x)-\int_{\mX} \bar K_{t+1}^{\alpha}( x,\cdot)\bar \pi_{t}^{\alpha}(\sd x).
\eeq
Adding and subtracting the term $ \int_{\mX} \bar K_{t+1}^\alpha(x,\cdot)\pi_{t}(\sd x)$ on the right hand side of \eqref{1} yields 

\begin{align*}
   \xi_{t+1}(\cdot)-\bar \xi_{t+1}^{\alpha}(\cdot)&=\int_{\mX} \left[K_{t+1}(x,\cdot)-\bar K_{t+1}^{\alpha}(x,\cdot)\right]\pi_{t}(\sd x)+\int_{\mX} \bar K_{t+1}^\alpha(x,\cdot)(\pi_{t}-\bar\pi_{t}^{\alpha})(\sd x)\\
    &=\int_{\mX} \left[K_{t+1}(x,\cdot)-\bar K_{t+1}^{\alpha}(x,\cdot)\right]\pi_{t}(\sd x)+\bar K_{t+1}^\alpha(\pi_{t}-\bar \pi_{t}^{\alpha}
    ).
\end{align*}
Applying the total variation norm, and  using Eq.\eqref{normK}  and Eq.\eqref{DKt},  we obtain
     \begin{align*}
 \hspace{-1.2cm}\norm{\xi_{t+1}-\bar \xi_{t+1}^{\alpha}}_{TV}& \leq \norm{\bar K_{t+1}^\alpha}_{\mathcal{D}_{\mX}}\norm{\pi_t-\bar \pi_t^{\alpha}}_{TV}+\Delta_{K_{t+1}}(\gamma),
\end{align*}
Now employing Proposition \ref{pimu} and subsequently Proposition \ref{rec}, yields
 \begin{align*}
    \hspace{1.5cm} \norm{\xi_{t+1}-\bar \xi_{t+1}^{\alpha}}_{TV}  &\leq \frac{\norm{\bar K_{t+1}^\alpha}_{\mathcal{D}_{\mX}}}{\xi_t(g_t)} \norm{\mu_t-\mu_t^{\alpha}}_{TV}+\Delta_{K_{t+1}}(\gamma) \\
     & \leq \frac{\norm{ K_{t+1}}_{\mathcal{D}_{\mX}}}{\xi_t(g_t)}\left[\norm{g_t}_{\infty} \norm{\xi_t-\bar \xi_t^{\alpha}}_{TV}+\Delta_{g_t}(\gamma)\right]+\Delta_{K_{t+1}}(\gamma).
    \end{align*}
Where we used Proposition \ref{DobN} to obtain the last inequality.
\end{proof}


\begin{lemma}\label{TT0}
    If Assumption \ref{A2} holds,  then 
    $$\lim_{\gamma\to 0} \Delta_{g_{t}}(\gamma)=0, \quad \lim_{\gamma\to 0}\Delta_{K_{t+1}}(\gamma)=0.$$
\end{lemma}
\begin{proof} 
For any $t\geq 1$, and a sequence $\{\gamma_n\}_{n\in \mbN}$ such that $\gamma_n\to 0$ when $n\to \infty$, define the sequences $\{h_n(x)\}_{n\in \mbN}$,  and $\{H_n(x)\}_{n\in\mbN}$ of real valued bounded functions  where 
$$h_n(x):=g_t(\alpha(x,\gamma_n))-g_t(x), \quad 
H_n(x):=\norm{K_{t+1}(x,\cdot)-K_{t+1}(\alpha_t(x,\gamma_n),\cdot)}_{TV},$$
(note that the t.v. norm is bounded for probability measures). The proof follows  from the continuity of the maps $\alpha_t(x,\gamma)$ w.r.t. $\gamma$. Indeed, by Assumption \ref{A2}, both $g_t$ and $K_{t+1}$ are continuous. Consequently, the functions $h_n(x)$ and $H_n(x)$ converge pointwise to zero. Furthermore, since these are sequences of bounded functions, we can apply the Lebesgue's dominated convergence Theorem to complete the proof.
\end{proof}

%
\subsection{Main proof}\label{B2}

The preliminary results in Section \ref{B1} provide us the elements to prove the key result in Lemma \ref{TI} below.
 
\begin{lemma}\label{TI}
    If Assumption \ref{A2} holds then, for any $t\in \mbN$ finite and for any $\epsilon>0$ there exists a sequence of parameters $\gamma_{0:t}(\epsilon)$ such that 
    
\beq \label{E2}
\Delta_{g_i}(\gamma_i)\geq \norm{g_i}_{\infty}\norm{\xi_i-\bar \xi_i^{\alpha}}_{TV}, \;\; i=1,...,t, 
        \eeq
         \beq \label{E1}
       \text{and} \quad  \norm{\xi_{t+1}-\bar \xi_{t+1}^{\alpha}}_{TV}\leq \epsilon.  
        \eeq
\end{lemma}
\begin{proof}\label{ISSLP} We proceed by induction, starting at t=1.
    Given $\epsilon>0$, using Lemma \ref{TT0} is possible to choose $\gamma_1> 0$ such that  
    $$a_1\Delta_{g_1}(\gamma_1)+\Delta_{K_2}(\gamma_1)\leq \epsilon,$$
 Note that, by Eq. \eqref{A1}, for any $\gamma_1> 0$ we have $$\Delta_{g_1}(\gamma_1)\geq \norm{g_1}_{\infty} \norm{\xi_1-\bar \xi_1^{\alpha}}_{TV}=0,$$ 
 wich implies, by Corrolary \ref{CNI}, that $\xi_1(g_1)\leq\bar \xi_1^\alpha(g_1^\alpha)$. Then, using Lemma \ref{pmrl},
    $$\norm{\xi_2-\bar\xi_2^{\alpha}}_{TV}\leq b_2{\norm{\xi_1-\bar\xi_1^{\alpha}}_{TV}}+a_1\Delta_{g_1}(\gamma_1)+\Delta_{K_{2}}(\gamma_1)\leq\epsilon.$$
   
    For the induction step, assume that at time $t-1$   there is a sequence $\gamma_{0:t-1}(\epsilon)$ such that \eqref{E2} and \eqref{E1} hold  for any  given $\epsilon>0$. 
    
    At time $t$, we use Lemma \ref{TT0} to choose $\gamma_t> 0$ such that
    $$a_t\Delta_{g_t}(\gamma_t)+\Delta_{K_{t+1}}(\gamma_t)\leq \frac{\epsilon}{2}.$$
    Moreover, by Eq. \eqref{A1}, we ensure that $\Delta_{g_t}(\gamma_t)>0,$ and define 
    $$\epsilon^\star:=\min\left\{\frac{\Delta_{g_t}(\gamma_t)}{\norm{g_t}_{\infty}},\frac{\epsilon}{2b_t} \right\}>0.$$
    Then, by the induction hypothesis, there is a sequence $\gamma_{0:t-1}(\epsilon^\star)$ such that \eqref{E2} and \eqref{E1} hold. The latter implies that
    $$\norm{\xi_{t}-\bar\xi_{t}^{\alpha}}_{TV}\leq \epsilon^\star,$$ therefore, by the definition of $\epsilon^\star$ we have $ \Delta_{g_t}(\gamma_t)\geq \norm{g_t}_{\infty} \norm{\xi_{t}-\bar \xi_{t}^{\alpha}}_{TV}$. 
    Moreover by \eqref{E2}
$$\Delta_{g_i}(\gamma_i)\geq \norm{g_i}_{\infty} \norm{\xi_i-\bar\xi_i^{\alpha}}_{TV}, \;\; i=1,...,t-1.$$
     Now, using Lemma \ref{pmrl}, and by construction
   $$ \norm{\xi_{t+1}-\bar\xi_{t+1}^{\alpha}}_{TV} \leq b_t\norm{\xi_{t}-\bar \xi_{t}^{\alpha}}_{TV}+ a_t\Delta_{g_t}(\gamma_t)+\Delta_{K_{t+1}}(\gamma_t)\leq \epsilon.$$
\end{proof}
We can finally proceed with the proof of Theorem \ref{ISSL}.
\begin{proof}
 By Lemma \ref{TI}, for any $T\geq 1$, and  $\epsilon>0$ there is a sequence $\gamma_{1:T}(\epsilon)$ such that 
$$\Delta_{g_i}(\gamma_i)\geq \norm{g_i}_{\infty}\norm{\xi_i-\bar \xi_i^{\alpha}}_{TV}, \;\; i=1,...,T, $$
$$\text{and} \quad  \norm{\xi_{T+1}-\bar \xi_{T+1}^{\alpha}}_{TV}\leq \epsilon.$$
And by Corollary \ref{CNI} we have $\bar\xi_t^\alpha(g_t)\geq \xi_t(g_t),$ for $t=1,...,T.$ Therefore  \beq\label{PF}
  p_T(y_{1:T}\vert \bar\mM^\alpha)=\prod_{i=1}^T\bar\xi_i^\alpha(g_i^\alpha) \geq \prod_{i=1}^T\xi_i(g_i)=p_T(y_{1:T}\vert \mM).
  \eeq
   Finally, using  Remark \ref{SBER}, equation \eqref{PF}  implies that $$p_T(y_{1:T}\vert \mM^{\alpha})\geq  p_T(y_{1:T}\vert \mM). $$
\end{proof}

%
\section{Lipschitz parametric models}\label{App3}

We rely on the following result.

\begin{lemma}\label{PPL}
Let Assumption \ref{A3} ii), iii) hold and also assume that at some time $t\geq 1$ we have $\norm{\pi_{t-1,\theta}-\pi_{t-1,\theta'}}_{TV}\leq C_{\pi_{t-1}} \abs{\theta-\theta'}$ for some constant $C_{\pi_{t-1}}\in \mbR^+$. Then
\begin{enumerate}
    \item $\norm{\xi_{t,\theta}-\xi_{t,\theta'}}_{TV}
      \leq C_{K_t} \abs{\theta-\theta'}.$\label{pmrL}
    \item $\norm{\mu_{t,\theta}-\mu_{t,\theta'}}_{TV}\leq C_{\mu_t} \abs{\theta-\theta'}$, where $\mu_{t,\theta}$ is given in Definition \ref{UNM}.\label{pimuL}
    \item $\norm{\pi_{t,\theta}-\pi_{t,\theta'}}_{TV}\leq C_{\pi_{t}} \abs{\theta-\theta'},$\label{recL}
\end{enumerate}
for some constants $C_{K_{t}},C_{\mu_{t}},C_{\pi_{t}}\in \mbR^+$.
\end{lemma}
\begin{proof}
The proof of 1. follows the same argument as the proof of Proposition \ref{pmrl}, the proof of 2. is the same as the proof of Lemma \ref{rec}, and the proof of 3. follows the same argument as the proof of  Proposition \ref{pimu}.
\end{proof}

\begin{corollary}\label{PPLC}
  If Assumption \ref{A3} holds, then statements 1., 2., and 3. in Lemma \ref{PPL} hold for any $t\geq 1.$
\end{corollary}

This result implies that the normalisation constants $\xi_{t,\theta}(g_{t,\theta})$ are Lipschitz w.r.t the parameter $\theta$ for any time step $t$. Indeed, if at time $t$ we have that  $\norm{\mu_{t,\theta}-\mu_{t,\theta'}}_{TV}\leq C_{\mu_t} \abs{\theta-\theta'}$ then, in particular, its normalisation constants satisfy $\abs{\xi_{t,\theta}(g_{t,\theta})-\xi_{t,\theta'}(g_{t,\theta'})}\leq C_{\mu_t} \abs{\theta-\theta'}$. Therefore, the following result is straightforward

\begin{theorem}\label{thm:Lipchitzthm}
   Let $p_T(y_{1:T}\vert \mM_{\theta})=\prod_{i=1}^T \xi_{i,\theta}(g_{i,\theta})$ be the Bayesian evidence of the parametric model $\mM_\theta$ at time $T$.  If Assumption \ref{A3} holds, then the Bayesian evidence is Lipschitz w.r.t. the parameter $\theta$, i.e.,
    \beq
\abs{p_T(y_{1:T}\vert \mM_{\theta})-p_T(y_{1:T}\vert \mM_{\theta'})} \leq L_T\abs{\theta-\theta'}, \quad L_T \in \mbR^+,
    \eeq 
\end{theorem}
\begin{proof} For $T=1$ we have $$\abs{p_1(y_{1}\vert \mM_{\theta})-p_1(y_{1}\vert \mM_{\theta'})}=\abs{\xi_{1,\theta}(g_{1,\theta})-\xi_{1,\theta'}(g_{1,\theta'})}\leq \norm{\mu_{1,\theta}-\mu_{1,\theta'}}_{TV}.$$
Using Corollary  \ref{PPLC} (Lemma \ref{PPL} \ref{pimuL}.) we get the result for $T=1$. 

For the induction step, assume that at time $T-1$ we have $\norm{\pi_{T-1,\theta}-\pi_{T-1,\theta'}}_{TV}\leq  C_{\pi_{T-1}} \abs{\theta-\theta'},$ then by Corollary \ref{PPLC} (Lemma \ref{PPL} \ref{pimuL}.)  $\norm{\mu_{T,\theta}-\mu_{T,\theta}}_{TV}\leq C_{\mu_{T}}\abs{\theta-\theta'}$, therefore, in particular,  $\abs{\xi_{T,\theta}(g_{T,\theta})-\xi_{T,\theta'}(g_{T,\theta'})}\leq C_{\mu_{T}}\abs{\theta-\theta'}$. (Note that we can propagate this to the next step by Lemma \ref{PPL} \ref{recL}. i.e.  $\norm{\pi_{T,\theta}-\pi_{T,\theta'}}_{TV}\leq  C_{\pi_T} \abs{\theta-\theta'}$).

Since the product of a finite number of bounded Lipchitz functions is again Lipchitz, we have that 
$p_T(y_{1:T}\vert \mM_{\theta})=\prod_{i=1}^T \xi_{i,\theta}(g_{i,\theta})$ is a Lipschitz function w.r.t. the parameter $\theta$.
\end{proof}

%
\section{Error between models}\label{app:proof_filters}

Let $\varphi : \mathcal{X}^{\otimes T} \to \mathbb{R}$ be a bounded test function and $\Pi_T^{\theta^\star}(X_{1} \in \sd x_1, \ldots, X_T \in \sd x_T) := \mathbb{P}_{\theta^\star}(X_1 \in \sd x_1, \ldots, X_T \in \sd x_T | Y_{1:T} = y_{1:T})$ be the filter with the MLE estimate $\theta^\star$ and $\Pi_T^{\theta, \alpha}$ denote the corresponding filter with the nudged model. Note that, we can write
\begin{align*}
\Pi_T^{\theta^\star}(\varphi) = \frac{\mathsf{p}^{\theta^\star}_{0:T}(\varphi g_{1:T})}{\mathsf{p}_{0:T}^{\theta^\star}(g_{1:T})}, 
\end{align*}
where $g_{1:T} := g_1 \times \cdots \times g_T$ is the product of likelihoods and
\begin{align*}
\mathsf{p}_{0:T}^{\theta^\star}(\sd x_{0:T}) = \pi_0(\sd x_0) \prod_{t=1}^T K_{t,\theta^\star}(\sd x_t | x_{t-1}).
\end{align*}
Note also that $\mathsf{p}_{0:T}^{\theta^\star}(g_{1:T}) = p_T(y_{1:T} | \mathcal{M}_{\theta^\star})$ which will be of use later. A similar representation holds for the nudged kernel, i.e.,
\begin{align*}
\Pi_T^{\theta, \alpha}(\varphi) = \frac{\mathsf{p}^{\theta,\alpha}_{0:T}(\varphi g_{1:T})}{\mathsf{p}^{\theta,\alpha}_{0:T}(g_{1:T})},
\end{align*}
where
\begin{align*}
\mathsf{p}_{0:T}^{\theta, \alpha}(\sd x_{0:T}) = \pi_0(\sd x_0) \prod_{t=1}^T K^\alpha_{t,\theta}(\sd x_t | x_{t-1}).
\end{align*}
Some straightforward manipulations complete the analysis,
\begin{align*}
&\left| \Pi_T^{\theta^\star}(\varphi) - \Pi_T^{\theta, \alpha}(\varphi) \right| = \left| \frac{\mathsf{p}^{\theta^\star}_{0:T}(\varphi g_{1:T})}{\mathsf{p}_{0:T}^{\theta^\star}(g_{1:T})} - \frac{\mathsf{p}^{\theta,\alpha}_{0:T}(\varphi g_{1:T})}{\mathsf{p}^{\theta,\alpha}_{0:T}(g_{1:T})}, \right| \\
&\leq \left| \frac{\mathsf{p}^{\theta^\star}_{0:T}(\varphi g_{1:T})}{\mathsf{p}_{0:T}^{\theta^\star}(g_{1:T})} - \frac{ \mathsf{p}^{\theta,\alpha}_{0:T}(\varphi g_{1:T})}{\mathsf{p}^{\theta^\star}_{0:T}(g_{1:T})}\right| + \left| \frac{ \mathsf{p}^{\theta,\alpha}_{0:T}(\varphi g_{1:T})}{\mathsf{p}^{\theta^\star}_{0:T}(g_{1:T})} - \frac{\mathsf{p}^{\theta,\alpha}_{0:T}(\varphi g_{1:T})}{\mathsf{p}^{\theta,\alpha}_{0:T}(g_{1:T})} \right| \\
&\leq \frac{\|\varphi\|_\infty | p_T(y_{1:T} | \mathcal{M}_{\theta^\star}) - p_T(y_{1:T} | \mathcal{M}_\theta^\alpha)|}{p_T(y_{1:T} | \mathcal{M}_{\theta^\star})} + \|\varphi\|_\infty p_{T}(y_{1:T} | \mathcal{M}_\theta^\alpha) \frac{| p_T(y_{1:T} | \mathcal{M}_{\theta^\star}) - p_T(y_{1:T} | \mathcal{M}_\theta^\alpha)|}{|p_T(y_{1:T} | \mathcal{M}_{\theta^\star}) p_T(y_{1:T} | \mathcal{M}_\theta^\alpha)|} \\
&= \frac{2\|\varphi\|_\infty | p_T(y_{1:T} | \mathcal{M}_{\theta^\star}) - p_T(y_{1:T} | \mathcal{M}_\theta^\alpha)|}{p_T(y_{1:T} | \mathcal{M}_{\theta^\star})}.
\end{align*}

%
\section{Projected gradient-ascent nudging}\label{A:PGAN}


Let us define the projection operator onto the set \( \mathcal{X} \subset \mbR^{d_x} \) as  
\beq
P_\mX \left( z \right):=\arg\min_{y \in \mX} \left\| y - z \right\|^2, \; \; \forall z\in \mbR^{d_x}.
\label{eqProj}
\eeq
Clearly, $P_\mX:\mX\mapsto\mX$. If we choose a differentiable function $f:\mX \mapsto \mathbb{R}$ then, from the latter definition \eqref{eqProj}, we can construct a projected gradient ascent (PGA) step, of the form
\[
P_\mX \left( x + \nabla f(x) \right) = \arg\min_{y \in \mX} \left\| y - (x+\nabla f(x)) \right\|^2,
\]
which extends the notion of gradient ascent for constrained optimisation \cite{nesterov2013introductory}. Furthermore, let us introduce the PGA operator 
\[
F(x) := P_{\mathcal{X}}(x + \nabla f(x)) - x.
\]
Note that if \( x + \nabla f(x) \in \mathcal{X} \) then we recover the standard gradient \( F(x) = \nabla f(x) \).


The following proposition is a minor variation of Lemma 1.2.3 in \cite{nesterov2013introductory} that plays a fundamental role in our analysis.  
\begin{proposition}
For any \( x, y \in \mX \), if \( \nabla f (x) \) is \( L \)-Lipschitz, then  
    \begin{equation}\label{eq:P1}
         f(y) \geq f(x) + (\nabla^\top f(x)) (y - x) - \frac{L}{2} \| y - x \|^2,
    \end{equation}
where the superscript ${^\top}$ denotes transposition. 
\end{proposition}

Next, we obtain a result for the PGA operator $F(x)$ that is analogous to Proposition~\ref{IGA}.
 \begin{lemma} Assume that $\mX$ is a closed and convex set and let the function $f:\mX \mapsto \mbR$ be differentiable, with $L$-Lipschitz continuous gradient $\nabla f$. Then for all $x\in \mX$ such that  $F(x)\neq 0$, we have     
        \begin{equation}\label{eq:PGA}
                    f(x + \gamma F(x)) \geq f(x) + \gamma \left(1 - \frac{\gamma L}{2} \right) \|F(x)\|^2 > f(x), \; \forall \gamma \in (0,2/L).
        \end{equation}
    \end{lemma}
\begin{proof}
Define \( y = x + \gamma F(x) \). Then, from \eqref{eq:P1} we obtain the inequality
\begin{align*}
f(y) &\geq f(x) + (\nabla^\top f(x)) (y - x) - \frac{L}{2} \|y - x\|^2  \\
&= f(x) + \gamma (\nabla^\top f(x)) F(x) - \frac{\gamma^2 L}{2} \|F(x)\|^2.
\end{align*}
Adding and subtracting \( \gamma \|F(x)\|^2 = \gamma F(x)^\top F(x) \) in the expression above, we get
        \[
        f(y) \geq f(x) + \gamma (\nabla f(x) - F(x))^\top F(x) + \gamma\left( 1 -       \frac{\gamma L}{2} 
        \right) \|F(x)\|^2.
        \]
To complete the proof, it is sufficient to show that  
    \[
    \gamma (\nabla f(x) - F(x))^\top F(x) \geq 0.
    \]
To see this, recall that by the definition of \( F(x) \),
    \begin{align}
        \gamma (\nabla f(x) - F(x)) &= \gamma (x + \nabla f(x) - P_{\mathcal{X}}(x + \nabla f(x))) \notag \\
        &= \gamma (z - P_{\mathcal{X}}(z)), \notag
    \end{align}
where we have defined \( z := x + \nabla f(x) \). Then, using the minimum principle of the Euclidean projection,
\begin{equation}\label{eq:Min_Pri}
(P_{\mathcal{X}}(z) - z)^\top (y - P_{\mathcal{X}}(z)) \geq 0 \quad \forall y,z\in\mX,
\nonumber
\end{equation}
we readily obtain
    \[
    \gamma (\nabla f(x) - F(x))^\top F(x) = \gamma ( P_{\mathcal{X}}(z) - z)^\top ( x - P_{\mathcal{X}}(z) ) \geq 0,
    \]
    which concludes the proof.
\end{proof}

Finally, the same as in Section \ref{SGAN}, we use the gradient of \( \log g_t \) to nudge the Markov kernel \( K_t \) towards regions of the state space \( \mathcal{X} \) where the likelihood is higher. Assume that \( \mathcal{X} \) is closed and convex and define the projected nudging transformation  
\begin{equation}\label{eq:PGANT}
\alpha_t(x, \gamma) = x + \gamma G_t(x),    \;\; \forall x \in \mX, \; \gamma \in [0,2/L],
\end{equation}
where \( G_t(x) := P_{\mathcal{X}}\left(x + \nabla \log g_t(x)\right) - x \) and \( \gamma \) is the usual step-size parameter.
\begin{remark}
Note that \( \alpha_t(x, \gamma) = x + \gamma G_t(x) \)  effectively nudges \( x \) towards the projected log-gradient direction, and $\alpha_t(x, \gamma) \in \mathcal{X}$. In fact, it is straightforward to see that 
\[
\alpha_t(x, \gamma) = (1 - \gamma)x + \gamma P_{\mathcal{X}}(x + \nabla \log g_t(x)), 
\] 
and we recall that \( \mathcal{X} \) is assumed to be convex, then $ \alpha_t(x, \gamma) \in \mathcal{X}, \; \forall \gamma \in [0,1]$. 
\end{remark}

From \eqref{eq:PGANT}, it follows that the projected nudging transformation is continuous with respect to the parameter \( \gamma \). Using \eqref{eq:PGA}, we can easily derive the analogous result to Eq.~\eqref{eq:ExponentialIncreaseg} for this scenario. Therefore, if there exists a set \( A_t \subseteq \mathcal{X} \) such that \( G_t(x) \neq 0 \) for all \( x \in A_t \) and \( \xi_t(A_t) > 0 \), the analogous result to Eq.~\eqref{eq:Deltag} also holds for the PGA nudging \eqref{eq:PGANT}. We are now in a position to state the following result, which is analogous to Corollary~\ref{cor:GANS}. 
\begin{corollary}
    For $t=1,..,T$, let $\alpha_t$ have the form in \eqref{eq:PGANT} and let $Y_{1:T}=y_{1:T}$ be an arbitrary but fixed data set. Assume that the state space $\mX$ is closed and convex. If 
\begin{enumerate}[(a)]
\item $\nabla \log g_t(x)$ is $L_t$-Lipschitz continuous,
\item there are sets $A_t \subseteq \mX$ such that $\xi_t(A_t)>0$  and $G(x)\neq0$ for all $x \in A_t$, and
\item Assumption \ref{A2} holds,
\end{enumerate}
then there exists a positive sequence $\gamma_{0:T}$ (depending on $\mM$ and $y_{1:T}$) such that
\beq
p_T(y_{1:T}|\mM^\alpha) \ge p_T(y_{1:T}|\mM).
\nn
\eeq
\end{corollary}

\end{appendices}

 
\bibstyle{sn-mathphys-num}
\bibliography{sn-bibliography}

\end{document}